\documentclass[journal,onecolumn,twoside]{IEEEtran}
\normalsize
\ifCLASSINFOpdf
\else
\fi
\usepackage{setspace}
\usepackage{color}
\usepackage{cite}
\usepackage{amsmath}
\usepackage{amsfonts}
\usepackage{amssymb}
\usepackage{amsthm}
\usepackage{tikz}
\usepackage{cases}
\usepackage{bm}
\usepackage{graphicx}
    \graphicspath{{../}}
    \DeclareGraphicsExtensions{.pdf}
\usepackage[caption=true,font=footnotesize]{subfig}
\usepackage{multirow}
\usepackage{makecell}
\usepackage{mathdots}
\usepackage{booktabs}
\usepackage{url}
\usepackage{bm}
\usepackage{xtab}
\usepackage{pifont}
\usepackage{array}
\usepackage{algorithm}
\usepackage{algorithmic}
\usepackage{enumerate}
\usetikzlibrary{arrows}
\usepackage{caption}
\usepackage{soul}
\usepackage{xcolor}

\allowdisplaybreaks[4]
\usepackage[colorlinks,
            linkcolor=blue,
            anchorcolor=blue,
            citecolor=blue]{hyperref}
\theoremstyle{plain}

\newtheorem{theorem}{Theorem}
\newtheorem{lemma}[theorem]{Lemma}

\newtheorem{definition}[theorem]{Definition}
\newtheorem{corollary}[theorem]{Corollary}
\theoremstyle{definition}
\newtheorem{example}{Example}

\newtheorem{remark}{Remark}

\begin{document}

\title{On the Number of Subsequences in the Nonbinary Deletion Channel}

\author{Han Li,~Xiang Wang, and ~Fang-Wei Fu
\thanks{X. Wang is with the School of Mathematics, Statistics and Mechanics, Beijing University of Technology, Beijing, 100124, China (e-mail: xwang@bjut.edu.cn). Han Li is with the Chern Institute of Mathematics and LPMC, Nankai University, Tianjin 300071, China (e-mail: hli@mail.nankai.edu.cn). F.-W. Fu is with the Chern Institute of Mathematics and LPMC, Nankai University, Tianjin 300071, China (e-mail: fwfu@nankai.edu.cn).}}




\maketitle

\begin{abstract}
In the deletion channel, an important problem is to determine the number of subsequences derived from a string $U$ of length $n$ when subjected to $t$ deletions. It is well-known that the number of subsequences in this setting exhibits a strong dependence on the number of runs in the string $U$, where a run is defined as a maximal substring of identical characters. In this paper we study the number of subsequences of non-binary strings in this scenario, and propose some improved bounds on  the number of subsequences of $r$-run non-binary strings. Specifically,  we characterize a family of $r$-run non-binary strings with the 
maximum number of subsequences under any $t$ deletions, and show that this number can be computed in polynomial time.
\end{abstract}


\section{Introduction}
\IEEEPARstart{L}{et} $X\in \Sigma_q^n$ be a $q$-ary string with length $n$, and let $t\leq n$ be a parameter, where $\Sigma_q=\{0,1,...,q-1\}$. Let $\Sigma_q^n$ be the set of all length-$n$ $q$-ary sequences. Let $\Sigma_q^*$ be the set of all $q$-ary sequences, that is, $\Sigma_q^*=\bigcup\limits_{n=0}^{\infty} \Sigma_q^n$.  
A sequence $Y \in \Sigma_q^{n-t}$ is called a $t$-subsequence of $X$ if it is derived  from $X$ by deleting exactly $t$ symbols. The \emph{deletion ball} of radius $t$ centered at $X \in \Sigma_q^n$ is defined as the collection of $t$-subsequences of $X$, denoted by $D_t(X)$. Formally, for any $X\in \Sigma_q^n$, we define
\begin{equation}
D_t(X) = \{Y\in \Sigma_q^{n-t} | Y \text{ is a $t$-subsequence of } X\}.\nonumber
\end{equation}
In deletion channels, the set $D_t(X)$ and its size have been widely studied in the design and analysis of communication schemes \cite{Mitzenmacher, Mercier2, Kash, L0} and in the sequence reconstruction problem \cite{Zhang, Gabrys, L1, L2, Pham1, Pham2}. It’s challenging to analyze $D_t(X)$ because
$\left|D_t(X)\right|$  depends not only on its length $n$ and the number $t$ of deletions, but also strongly on its structure.  Moreover, the size of ball has also been studied for deletions and insertions channel \cite{Sun, Wang, Lan, Bar, He}.

In 1966, Levenshtein \cite{L0} first proved that $\left|D_t(X)\right|$ depends on the number of runs in the string $X$. A \emph{run} is a maximal substring consisting of identical symbols, and we denote the number of runs in a string $X \in \Sigma_q^n$ as $r(X)$. 
Specifically, Levenshtein \cite{L0} showed  that
\begin{align*}
\binom{r(X)-t+1}{t}\leq \left|D_t(X)\right|\leq \binom{r(X)+t-1}{t} 
\end{align*}
for any $X\in \Sigma_q^n$. Let $\bm{\sigma} = (\sigma_1, \dots, \sigma_q)$ be a permutation of $\Sigma_q$. We define the periodic sequence $C_q(n, \bm{\sigma}) = c_1 \dots c_n$ such that $c_i = \sigma_i$ for $1 \le i \le q$ and $c_i = c_{i-q}$ for $i > q$. Let $\mathcal{C}_q(n)$ be the set of all such sequences for all possible permutations $\bm{\sigma}$.  Calabi and Hartnett \cite{Calabi} improved the upper bound on $\left|D_t(X)\right|$ and showed that the maximal value of $\left|D_t(X)\right|$, denoted by $D_q(n,t)$, is attained by arbitrary strings of the form $C_q(n,\bm{\sigma}) \in \mathcal{C}_q(n)$. Moreover, \cite{Calabi} gave a recursive expression for $D_t(C_q(n,\bm{\sigma}))$, thereby obtaining the bound
\begin{align*}
\binom{r(X)-t+1}{t}\leq \left|D_t(X)\right|\leq \left|D_t(C_q(n,\bm{\sigma}))\right|=\sum\limits_{i=0}^{q-1} D_q(n-i-1,t-i)
\end{align*}
for any $X\in \Sigma_q^n$. Hirschberg and Regnier \cite{Hirschberg} proposed an improved lower bound and an explict upper bound on $\left|D_t(X)\right|$ as follows
\begin{align*}
\sum\limits_{i=0}^{t}\binom{r(X)-t}{i}\leq \left|D_t(X)\right|\leq \left|D_t(C_q(n,\bm{\sigma}))\right|=\sum\limits_{i=0}^{t}\binom{n-t}{i}D_{q-1}(t,t-i),
\end{align*}
where $D_1(n,t)=1$ if $n \geq t \geq 0$ and $D_q(n,t)=0$ otherwise. Throughout this paper, let $S(x_1, \dots, x_r; a_1, \dots, a_r)$ denote an $r$-run string, where the $i$-th run has length $x_i \ge 1$ and consists of the symbol $a_i \in \Sigma_q$ for $i \in [r]$. For instance, $S(1, 2, 3; 0, 1, 2)$ represents $011222$. We use the notation $n \times m$ to denote $n$ consecutive runs, each of length $m$. For example, $S(2, 3 \times 1, 3; 1, 3, 2,1,0)= 11321000$. Liron and Langberg \cite{Liron} refined the upper and lower bounds on $|D_t(X)|$ for any binary string $X \in \Sigma_2^n$ with $r$ runs as follows:
\begin{equation*}
|D_t(A_{r,n})| \le |D_t(X)| \le |D_t(B_{r,k})|,
\end{equation*}
where  $A_{r,n} = S((r-1)\times 1, n-r+1; 0,1,\dots, (r-1) \bmod 2)$ is the unbalanced string, and  $B_{r,k} = S(r\times k; 0,1, \dots, (r-1) \bmod 2)$ is the balanced string with uniform run lengths $k=n/r$ (assuming $r \mid n$).  Moreover, the authors derived both recursive expressions and closed-form expressions for these bounds. Mercier et al. \cite{Mercier1} studied the setting of small values for $t$ and presented explicit formulas for $\left|D_t(X)\right|$ for $t\leq 5$ and $X\in \Sigma_q^n$.

\subsection{Our main contributions}
In this paper, we generalize results in \cite{Liron} to the $q$-ary setting. 

To establish the lower bound, we employ certain string operations to map $X$ to a corresponding alternating binary string  that preserves the run-type structure of $X$. Crucially, we prove that these operations never increase the number of subsequences through the transformation. Consequently, $|D_t(X)|$ is lower-bounded by the number of subsequences of the resulting binary string, allowing us to leverage the known binary lower bound from \cite{Liron} to derive our result for the $q$-ary setting. 

Regarding the upper bound, specifically for the case where $r \mid n$, we identify a class of extremal strings, referred to as the $q$-ary balanced strings. Any such string, represented by $S(r \times k; c_1, \dots, c_r)$ with constant run lengths $k = n/r$ and a periodic sequence  $c_1 \dots c_r \in \mathcal{C}_q(r)$, is shown to maximize the number of subsequences. By extending the structural analysis in \cite{Liron}, we derive an explicit recursive expression and a corresponding closed-form formula for the number of subsequences within these balanced strings. Based on this characterization, we first establish a tight upper bound for $r \mid n$ in  Theorem \ref{thm2}, and subsequently derive a universal upper bound for arbitrary $n$ and $r$ in Corollary \ref{cor1}. Although the latter may not be tight in all scenarios, it significantly refines the best-known results in the existing literature for the general $q$-ary setting. Finally, based on the explicit expression for the number of subsequences in these balanced strings, we perform an asymptotic analysis by extending the method presented in \cite{Liron}.

\subsection{Organization of this Paper}

The structure of the remainder of the paper is as follows. In Section~\ref{sec2}, we recall our notation and  some foundational results. In Section~\ref{sec3}, we introduce several basic string operations and use the known results on the number of subsequences for binary strings to establish a lower bound for general $q$-ary strings. In Section \ref{sec4}, we provide a family of balanced strings $\subseteq \Sigma_q^n$  which have the largest number of subsequences under deletion for any given  number of runs $r$ and deletions $t$. In Section \ref{sec5},  we derive an upper bound on the number of subsequences  for any given $q,n$, number of runs $r$ and deletions $t$, and compare it with  previously best-known results. Section  \ref{sec6} concludes the paper.

\section{Definitions and Preliminaries}
\label{sec2}

We denote by  $[n]$ the set $\{1,2,...,n-1,n\}$. Let $S_q$ be the set of all permutations over $\Sigma_q$. We use lower-case letters to denote scalars and uppercase letters to denote vectors or sequences. Let $X=x_1x_2\ldots x_n\in \Sigma_q^n$. We let $X_{[i,j]}$  be the projection of $X$ onto the index interval $[i,j]$, i.e., $X_{[i,j]}=x_i\ldots x_j$. 
Let $\mathcal{C} \subset \Sigma_q^n$ be a set and let $U$ be a sequence of length at most $n$.  We denote by $\mathcal{C}^{U}$ the subset of $\mathcal{C}$ comprising all sequences  that start with $U$. Similarly,  $\mathcal{C}_{U}$ is the subset of $\mathcal{C}$ comprising all sequences  that end with $U$.
Furthermore, for two sequences $U_1$, $U_2$ of length at most $n$, we define $\mathcal{C}_{U_2}^{U_1}$ as the set of sequences in $\mathcal{C}$ that start with $U_1$ and end with $U_2$.  For two $q$-ary sequences $U=u_1\ldots u_m,V=v_1\ldots v_n$, we denote $UV=u_1\ldots u_m v_1\ldots v_n$ as the concatenation of $U$ and $V$.  For a sequence $U$ and a set $\mathcal{C}$, the set $U  \circ \mathcal{C}$ is defined by prepending $U$ to every sequence in $\mathcal{C}$,
\[
U  \circ \mathcal{C} =\{U C~|~C\in \mathcal{C}\}.
\]
Similarly, the set $\mathcal{C} \circ U$ is defined by appending $U$ to every sequence in $\mathcal{C}$.

The total number of subsequences can be determined by partitioning the subsequence set into disjoint subsets based on their unique prefixes or suffixes. Since the cardinality of a set is additive over its mutually disjoint components, the structural properties established in Lemmas 1 and 2 of \cite{Gabrys} for binary strings naturally extend to general $q$-ary sequences for any $q \ge 2$, as formally stated in the following lemmas.

\begin{lemma}
\label{lm1}
Let $n,m_1,m_2,q$ be positive integers such that $m_1+m_2\leq n-t$. Then, for any $U\in \Sigma_q^n$ and $q\geq2$ we have
\begin{equation}
\left|D_t(U)\right|=\sum\limits_{U_1\in \Sigma_q^{m_1}}\sum\limits_{U_2\in \Sigma_q^{m_2}}\left|D_t(U)_{U_2}^{U_1}\right|.\nonumber
\end{equation}
\end{lemma}

\begin{lemma}
\label{lm2}
Let $n,m_1,m_2,t,q$ be positive integers such that $m_1+m_2\leq n-t$. Let $U=u_1...u_n\in \Sigma_q^n, U_1\in \Sigma_q^{m_1},$ and $U_2\in \Sigma_q^{m_2}$ with $q\geq 2$.
Assume that $k_1$ is the smallest integer such that $U_1$ is a subsequence of $u_1...u_{k_1}$ and $k_2$ is the largest integer such that $U_2$ is a subsequence of $u_{k_2}...u_n$. If $k_1<k_2$, then
\begin{equation}
D_t(U)_{U_2}^{U_1}=U_1\circ D_{t^{*}}(u_{k_1+1}...u_{k_2-1})\circ U_2,\nonumber
\end{equation}
where $t^*=t-(k_1-m_1)-(n-k_2+1-m_2)$. In particular, $\left|D_t(U)_{U_2}^{U_1}\right|=\left|D_{t^{*}}(u_{k_1+1}...u_{k_2-1})\right|$.
\end{lemma}

\begin{remark}
When $k_1 = k_2 - 1$, we adopt the convention that $u_{k_1+1} \dots u_{k_2-1} = \epsilon$, where $\epsilon$ denotes the empty string. It follows that $|D_{t^*}(\epsilon)|$ is $1$ for $t^* = 0$ and $0$ for $t^* > 0$.
\end{remark}

\section{Unbalanced strings and our  lower bound}
\label{sec3}
In this section, we use the binary unbalanced strings defined in \cite{Liron} to establish lower bounds on the number of subsequences for general $q$-ary strings. To do this, we require some basic operations on strings. These operations can  transform one string into another while preserving the monotonicity of the number of subsequences (i.e. either increasing or decreasing it). First, we recall the following  basic operation from \cite{Hirschberg}.

\textbf{\textit{1) Insertion Operation:}} Hirschberg and Regnier \cite{Hirschberg} proved that inserting a symbol anywhere in the middle of a string cannot reduce the number of subsequences.

\begin{lemma}[Insertion Increases the Number of
Subsequences \cite{Hirschberg}]
\label{lm3}
For any $q$-ary strings $U,V$ and any $a\in \Sigma_q$, we have 
$\left|D_t(UV)\right|\leq \left|D_t(UaV)\right|$.    
\end{lemma}

\textbf{\textit{2) Deletion Chain Rule:}} By the definiton of $D_t(\cdot)$, we can directly obtain the following lemma.
\begin{lemma}
\label{lm4}
Let $t,t'$ be positive integers. If $U$ is a $q$-ary string and  $V\in D_t(U)$, then  $D_{t'}(V)\subseteq D_{t+t'}(U)$.   
\end{lemma}

\textbf{\textit{3) Permutation Operation:}} We show that the number of subsequences of a string is invariant under any permutation of the alphabet. Specifically, let $f \in S_q$ be a permutation on $\Sigma_q$. For any string $U = S(x_1, \dots, x_r; a_1, \dots, a_r) \in \Sigma_q^n$, we define the action of $f$ on $U$ as$$f(U) = S(x_1, \dots, x_r; f(a_1), \dots, f(a_r)).$$
\begin{lemma}[Permutation Keeps the Number of
Subsequences Unchanged]
\label{lm5}
For any $q$-ary string $U=S(x_1,\ldots,x_r;a_1,\ldots,a_r)$ and a permutation $f\in S_q$,  
$\left|D_t(U)\right|=\left|D_t(f(U))\right|$.    
\end{lemma}
\begin{proof}
If $Y\in D_t(U)$ then $f(Y)\in D_t(f(U))$. Moreover, given two distinct $X,Y\in D_t(U)$, we have $f(X)\neq f(Y)$. Thus, $\left|D_t(U)\right|\leq \left|D_t(f(U))\right|$. Similarly, we also have $\left|D_t(f(U))\right|\leq \left|D_t(U)\right|$. So, the lemma follows.    
\end{proof}

\textbf{\textit{4) Reduction Operation:}} We introduce a reduction operation and show that it does not increase the number of subsequences.
\begin{definition}[Reduction Operation] 
For a string $S(x_1, \dots, x_r; a_1, \dots, a_r)$, we define the reduction operation as the transformation that yields a binary string
\[S^d(x_1, \dots, x_r; a_1, \dots, a_r) \triangleq S(x_1, \dots, x_r; d_1, \dots, d_r),\] 
where $d_i = (i - 1) \bmod 2$ for $i \in [r]$.
\end{definition}

\begin{lemma}[Reduction Operation Decreases the Number of
Subsequences]
\label{lm6}
For any $q$-ary string $U=S(x_1,\ldots,x_r;a_1,\ldots,a_r)$, it holds that
$\left|D_t(U^d)\right|\leq \left|D_t(U)\right|$.    
\end{lemma}
\begin{proof}
We proceed by induction on $r$.

\emph{Base Cases:}
When $r=1$, we have $U=a_1^{x_1}$ and $U^d=0^{x_1}$, which immediately implies  $|D_t(U)|=|D_t(U^d)|$. 

When $r=2$, we  have $U=a_1^{x_1}a_2^{x_2}$ and $U^d=0^{x_1}1^{x_2}$ with $a_1\neq a_2$. 
By defining a permutation function $g$ over $\Sigma_q$ such that $g(a_1)=0, g(a_2)=1$, and $g(i)=i$ for $i \in \Sigma_q \setminus \{a_1, a_2\}$, it follows from Lemma \ref{lm5} that $|D_t(U)| = |D_t(U^d)|$ .

\emph{Inductive Step:} Assume the claim holds for  any $q$-ary string $U$ with  $r(U) = k$ for all $1 \leq k \leq m$, where $m \geq 2$. We now consider the case where $r(U)=m+1$. Let $U=S(x_1,\ldots,x_{m+1};a_1,\ldots,a_{m+1})$ and  $U^d=S(x_1,\ldots,x_{m+1};d_1,\ldots,d_{m+1})$, where $d_i= (i-1) \bmod 2$ for $i\in [m+1]$.

By Lemmas $\ref{lm1}$ and $\ref{lm2}$, it follows that 
\begin{align}
\left|D_t(U)\right|&=\sum\limits_{i\in \Sigma_q}\left|D_t(U)^{i}\right|=\left|D_t(U)^{a_1}\right|+\left|D_t(U)^{a_2}\right|+\sum\limits_{i\in \Sigma_q \backslash\{a_1,a_2\}}\left|D_t(U)^{i}\right|\nonumber\\ 
&=\left|D_t(S(x_1-1,\ldots,x_{m+1};a_1,\ldots,a_{m+1}))\right|+\left|D_{t-x_1}(S(x_2-1,\ldots,x_{m+1};a_2,\ldots,a_{m+1}))\right|\nonumber\\
&~~~~+\sum\limits_{i\in \Sigma_q \backslash\{a_1,a_2\}}\left|D_{t-x_1}(S(x_2,\ldots,x_{m+1};a_2,\ldots,a_{m+1}))^{i}\right|\nonumber\\
&\geq \left|D_t(S(x_1-1,\ldots,x_{m+1};a_1,\ldots,a_{m+1}))\right|+\left|D_{t-x_1}(S(x_2-1,\ldots,x_{m+1};a_2,\ldots,a_{m+1}))\right|.\label{eq7}
\end{align}
Similarly, for the reduced sequence $U^d$, we have  
\begin{align}
\left|D_t(U^d)\right|=&\sum\limits_{i\in \Sigma_2}\left|D_t(U^d)^{i}\right|=\left|D_t(U^d)^{0}\right|+\left|D_t(U^d)^{1}\right|\nonumber\\ 
=&\left|D_t(S(x_1-1,\ldots,x_{m+1};d_1,\ldots,d_{m+1}))\right|+\left|D_{t-x_1}(S(x_2-1,\ldots,x_{m+1};d_2,\ldots,d_{m+1}))\right|.\label{eq8}
\end{align}

To establish $|D_t(U^d)| \le |D_t(U)|$, we compare the corresponding terms in \eqref{eq7} and \eqref{eq8}. Regarding the second terms, since $r(S(x_2-1, \ldots,x_{m+1}; a_2, \ldots, a_{m+1})) \leq m$, the induction hypothesis and Lemma \ref{lm5} yield
\begin{align*}
\left|D_{t-x_1}(S(x_2-1,\ldots,x_{m+1};d_2,\ldots,d_{m+1}))\right|\leq \left|D_{t-x_1}(S(x_2-1,\ldots,x_{m+1};a_2,\ldots,a_{m+1}))\right|.
\end{align*}
For the first terms, we consider the value of $x_1$:
\begin{itemize}
    \item If $x_1=1$: By the inductive hypothesis,
    \begin{align*}
\left|D_t(S(x_2,\ldots,x_{m+1};d_2,\ldots,d_{m+1}))\right|&\overset{(a)}{=}\left|D_t(S(x_2,\ldots,x_{m+1};d_2',\ldots,d_{m+1}'))\right|\\
&\leq \left|D_t(S(x_2,\ldots,x_{m+1};a_2,\ldots,a_{m+1}))\right|,
\end{align*}
where $d_i'= i \bmod 2$ and Equality $(a)$ follows from Lemma $\ref{lm5}$. Combining  this with Eqs. \eqref{eq7} and \eqref{eq8}, we obtain $|D_t(U^d)| \leq |D_t(U)|$.
\item If $x_1=2$: Note that the first term $|D_t(S(x_1-1, \ldots, x_{m+1}; a_1, \ldots, a_{m+1}))|$ in \eqref{eq7} corresponds to a sequence where the first run has length $1$. By applying the argument established for the $x_1=1$ case to this term, it follows that $|D_t(U^d)| \leq |D_t(U)|$.
\item If $x_1 \geq 3$: The result follows by repeatedly applying the logic established in the previous cases.
\end{itemize}
Thus, $|D_t(U^d)| \leq |D_t(U)|$ holds for $r(U)=m+1$.

By the induction hypothesis, the inequality
$\left|D_t(U^d)\right|\leq \left|D_t(U)\right|$ holds for $r(U)\geq 1$. Therefore, the lemma follows.
\end{proof}
\begin{remark}
\label{rmk1}
By Lemma $\ref{lm6}$, for any string $S(x_1,\ldots, x_{r};a_1,\ldots,a_{r})\in \Sigma_q^n$, we have 
\[\left|D_t(S(x_1,\ldots,x_{r};a_1,\ldots,a_{r}))\right|\geq |D_t(S(x_1,
x_2,\ldots, x_{r};0,1,\ldots,(r-1) \bmod 2)|.\]  Thus, to establish a lower bound on $\left|D_t(S(x_1,\ldots, x_{r};a_1,\ldots,a_{r}))\right|$,
it suffices to  restrict our consideration to  binary strings of the form $S(x_1,x_2,\ldots, x_{r};0,1,\ldots,(r-1) \bmod 2)$. 
Based on \cite[Theorem 2]{Liron}, we obtain the following result.
\end{remark}
\begin{theorem}[Lower bound]
For any $q$-ary string $U=S(x_1,\ldots,x_r;a_1,\ldots,a_r)$ with $n=\sum_{i=1}^r x_i$, we have 
\[|D_t(U)|\geq |D_t(S((r-1)\times 1,n-r+1; 0, 1, \dots, (r-1) \bmod 2)|.\]    
\end{theorem}

\section{Balanced strings}
\label{sec4}
In this section we define a family of strings named $q$-ary \emph{balanced strings}.  A $q$-ary string is called balanced if all its runs have the same length and the symbols across successive runs form a cyclic shift of a fixed permutation of $\Sigma_q$.  The set of all such balanced strings with $r$ runs and each of length $k$ is denoted by $\mathcal{B}^{q}_{r,k}$.  Formally, 
\begin{align*}
\mathcal{B}_{r,k}^q=\{S(r\times k;c_1,\ldots,c_r) \mid c_1\ldots c_r\in \mathcal{C}_q(r)\}.    
\end{align*}  We shall prove that for any fixed $n$ and $r$ with $r \mid n$, any balanced string with length $n$ and $r$ runs has  the maximal number of subsequences under any number of deletions. To prepare for the proof, we first introduce several operations on strings.

\subsection{Cyclic operation}
 To maximize the number of subsequences in the $q$-ary setting, we introduce a cyclic operation on the run symbols. While a simple alternation suffices for the binary case \cite{Liron}, a $q$-ary alphabet requires this operation to ensure symbol distinctiveness across consecutive runs and enhance the distinguishability of subsequences resulting from deletions.

\begin{definition}[Cyclic Operation] For any $q$-ary string $U = S(x_1, \dots, x_r; a_1, \dots, a_r)$, the \emph{cyclic operation} transforms $U$ into its canonical cyclic string $U^c$ by mapping its symbol sequence to a periodic pattern $c_1 \dots c_r \in \mathcal{C}_q(r)$ such that:
\begin{equation*}
U^c \triangleq S(x_1, \dots, x_r; c_1, \dots, c_r),
\end{equation*}
where $c_i = (a_1 + i - 1) \bmod q$ for all $i \in [r]$.
\end{definition}
 
\begin{lemma}[Cyclic Operation Increases the Number of
Subsequences]
\label{lm7}
For any $q$-ary string $U=S(x_1,\ldots,x_r;a_1,\ldots,a_r)$,  
$\left|D_t(U)\right|\leq \left|D_t(U^c)\right|$.  
\end{lemma}
\begin{proof}
We proceed by induction on $r$. 

\emph{Base Case}: When $r=1$, we have $U=a_1^{x_1}$ and $U^c=a_1^{x_1}$, which trivially implies $\left|D_t(U)\right|=\left|D_t(U^c)\right|$. Thus, the base case holds.

\emph{Inductive Step:}  Assume the statement holds for every $U$ with $r(U)\leq m$, where $m\geq 1$.

Consider $r(U)= m+1$ and write
\[U=S(x_1,\ldots,x_{m+1};a_1,\ldots,a_{m+1}),U^c=S(x_1,\ldots,x_{m+1};c_1,\ldots,c_{m+1}), \]
where
$c_i= (a_1+i-1) \bmod q$ for $i\in [m+1]$.
For $m+1 \leq q-1$, we extend the definition to $c_i = (a_1 + i - 1) \bmod q$ for $i \in [m+2, q]$ to cover the entire alphabet $\Sigma_q$ esuring that $\{c_i \mid i\in [q]\}=\Sigma_q$.  Define $f: \mathbb{N}\rightarrow \mathbb{N}$ as $f(i)=\sum_{j=1}^{i-1}x_{j}$ for $i\in [2,m+1]$. For each symbol $c_i$ with $i\in [q]$, let $l_i$ denote the index of its first occurrence in the string $a_1\dots a_{m+1}$. If $c_i$ does not appear, set $l_i=m+2$. Note that for any $i \neq j$ where $l_i, l_j\neq m+2$, we have $l_i\neq l_j$. 
By Lemmas $\ref{lm1}$ and $\ref{lm2}$, it follows that 
\begin{align*}
\left|D_t(U)\right|=&\sum\limits_{c_i\in \Sigma_q}\left|D_t(U)^{c_i}\right|=\left|D_t(U)^{c_1}\right|+\sum\limits_{i=2}^{q}\left|D_t(U)^{c_i}\right|\nonumber\\ 
=&\left|D_t(S(x_1-1,\ldots,x_{m+1};a_1,\ldots,a_{m+1}))\right|+\sum\limits_{i=2}^{q}\left|D_{t-x_1}(S(x_2,\ldots,x_{m+1};a_2,\ldots,a_{m+1}))^{c_i}\right|,
\end{align*}
where $\left|D_{t-x_1}(S(x_2,\ldots,x_{m+1};a_2,\ldots,a_{m+1}))^{c_i}\right|=\left|D_{t-f(l_i)}(S(x_{l_i}-1,\ldots,x_{m+1};a_{l_i},\ldots,a_{m+1}))\right|$ for $i\in [2,q]$. Let\\ $\hat{l}_1 \hat{l}_2 \dots \hat{l}_q$ denote the sequence $l_1 l_2 \dots l_q$ sorted in non-decreasing order, such that $\hat{l}_1 \leq \hat{l}_2 \leq \dots \leq \hat{l}_q$. Then
\begin{align}
\left|D_t(U)\right|=\left|D_t(S(x_1-1,\ldots,x_{m+1};a_1,\ldots,a_{m+1}))\right|+\sum_{i=2}^{q}|D_{t-f(\hat{l}_i)}(S(x_{\hat{l}_i}-1,\ldots,x_{m+1};a_{\hat{l}_i},\ldots,a_{m+1}))|.\label{eq10}
\end{align}
Similarly, let $l'_i$ denote the index of its first occurrence of $c_i$ in the string $c_1\dots c_{m+1}$ for $i \le \min\{m+1,q\}$. If $i > \min\{m+1,q\}$, we set $l'_i=m+2$. By Lemmas $\ref{lm1}$ and $\ref{lm2}$, we have  
\begin{align}
\left|D_t(U^c)\right|=&\sum\limits_{c_i\in \Sigma_q}\left|D_t(U^c)^{c_i}\right|=\left|D_t(U^c)^{c_1}\right|+\sum\limits_{i=2}^q\left|D_t(U^c)^{c_i}\right|\nonumber\\ 
=&\left|D_t(S(x_1-1,\ldots,x_{m+1};c_1,\ldots,c_{m+1}))\right|+\sum\limits_{i=2}^{q}\left|D_{t-x_1}(S(x_2,\ldots,x_{m+1};c_2,\ldots,c_{m+1}))^{c_i}\right|\nonumber\\
=&\left|D_t(S(x_1-1,\ldots,x_{m+1};c_1,\ldots,c_{m+1}))\right|+\sum\limits_{i=2}^{q}\left|D_{t-f(l'_i)}(S(x_{l'_i}-1,\ldots,x_{m+1};c_{l'_i},\ldots,c_{m+1}))\right|.\label{eq11}
\end{align}

Next we compare the  corresponding terms in \eqref{eq10} and \eqref{eq11}. By the definitions of $\hat{l}_i$ and $l'_i$, we have $\hat{l}_i\geq l'_i$ for $i\in [2,q]$. For the summands,  Lemma $\ref{lm4}$ implies for each $i\in [2,q]$,  
\begin{align*}
\left|D_{t-f(l'_i)}(S(x_{l'_i}-1,\ldots,x_{m+1};c_{l'_i},\ldots,c_{m+1}))\right|\geq& \left|D_{t-f(\hat{l}_i)}(S(x_{\hat{l}_i}-1,\ldots,x_{m+1};c_{\hat{l}_i},\ldots,c_{m+1}))\right|\\
\overset{(a)}{\geq} &\left|D_{t-f(\hat{l}_i)}(S(x_{\hat{l}_i}-1,\ldots,x_{m+1};a_{\hat{l}_i},\ldots,a_{m+1}))\right|,
\end{align*}
where Inequality $(a)$ follows from the induction hypothesis under the condition of $r(S(x_{\hat{l}_i}-1,\ldots,x_{m+1};a_{\hat{l}_i},\ldots,a_{m+1}))\leq m$.
We now examine the first terms of \eqref{eq10} and \eqref{eq11} based on the value of $x_1$:
\begin{itemize}
    \item When $x_1=1$, by the induction hypothesis,
\begin{align*}
\left|D_t(S(x_2,\ldots,x_{m+1};c_2,\ldots,c_{m+1}))\right|\geq \left|D_t(S(x_2,\ldots,x_{m+1};a_2,\ldots,a_{m+1}))\right|.
\end{align*}
So, combining with Eqs. $(\ref{eq10})$ and $(\ref{eq11})$, it follows that $\left|D_t(U)\right|\leq \left|D_t(U^c)\right|.$
\item When $x_1=2$, the first term in \eqref{eq10} corresponds to a sequence with a leading run of length $1$. Applying the argument for the $x_1=1$ case, it follows that $\left|D_t(U)\right|\leq \left|D_t(U^c)\right|$. 
\item When $x_1\geq 3$, the result follows by repeatedly applying the logic established in the previous cases.
\end{itemize}
 So, when $r(U)=m+1$, we have  $\left|D_t(U)\right|\leq \left|D_t(U^c)\right|$. 

By the induction, the lemma is proved.
\end{proof}
\begin{remark}
\label{rmk2}
By Lemma \ref{lm6}, the $t$-deletion ball size of any $q$-ary string is upper-bounded by that of its canonical cyclic string $U^c$. Since Lemma \ref{lm5} implies $|D_t(U^c)|$ is invariant for all $a_1 \in \Sigma_q$, we fix $a_1=0$ without loss of generality. Hereafter, we omit the symbol sequence and denote $U^c$ simply as:
\begin{equation*}
S(x_1, \dots, x_r) \triangleq S(x_1, \dots, x_r; 0, 1, \dots, (r-1) \bmod q).
\end{equation*}
This convention follows the analytical framework in \cite{Liron}, allowing us to focus exclusively on the optimization of run lengths $(x_1, \dots, x_r)$.
\end{remark}

Next, we establish several properties of $S(x_1, x_2, \dots, x_r)$ in the following lemmas, which generalize Lemma 3 in \cite{Liron}. Henceforth, for any string $U=S(x_1,\ldots,x_r)\in \Sigma_q^*$, let $q_1=\min\{r,q\}$. Unless otherwise specified, we consider the  alphabet $\Sigma_{q_1}$ for the string $U$. For a predicate $P$, we denote by $a|_P$ the value $a$ if $P$ holds, and $0$ otherwise. 

\begin{lemma}
\label{lm8}
Given any string $S(x_1,\ldots,x_r)\in \Sigma_q^*$, let $\sum_{i=1}^{r}x_i=n > t\geq 1$. Define $f(j) = \sum_{i=1}^j x_i$ and $g(j) = \sum_{i=1}^j x_{r-i+1}$ for $j \in [q_1-1]$. Then
we have 
\begin{align*}
\left|D_t(S(x_1,\ldots,x_r))\right|&=|D_t(S(x_2,\ldots,x_r))|+
\sum\limits_{i=0}^{x_1-1}\sum\limits_{j=1}^{q_1-1}|D_{t-f(j)+i}(S(x_{j+1}-1,\ldots,x_{r}))|+1|_{t>n-x_1},\\
\left|D_t(S(x_1,\ldots,x_r))\right|&=|D_t(S(x_1,\ldots,x_{r-1}))|+
\sum\limits_{i=0}^{x_r-1}\sum\limits_{j=1}^{q_1-1}|D_{t-g(j)+i}(S(x_1,\ldots,x_{r-j}-1))|+1|_{t>n-x_r}.
\end{align*}
\end{lemma}
\begin{proof}
By applying Lemma $\ref{lm1}$ once, we have 
\begin{align*}
\left|D_t(S(x_1,\ldots,x_r))\right|=|D_t(S(x_1-1,x_2,\ldots,x_r))|+
\sum\limits_{j=1}^{q_1-1}|D_{t-f(j)}(S(x_{j+1}-1,\ldots,x_{r}))|.   
\end{align*}
If $x_1>1$ then we can use Lemma $\ref{lm1}$ again and get
\begin{align*}
\left|D_t(S(x_1,\ldots,x_r))\right|=|D_t(S(x_1-2,x_2,\ldots,x_r))|+\sum\limits_{i=0}^{1}
\sum\limits_{j=1}^{q_1-1}|D_{t-f(j)+i}(S(x_{j+1}-1,\ldots,x_{r}))|.  
\end{align*}
Similarly, if $s\leq \min\{x_1,n-t\}$ then we can apply Lemma $\ref{lm1}$ $s$ times to obtain that 
\begin{align*}
\left|D_t(S(x_1,\ldots,x_r))\right|=|D_t(S(x_1-s,x_2,\ldots,x_r))|+\sum\limits_{i=0}^{s-1}
\sum\limits_{j=1}^{q_1-1}|D_{t-f(j)+i}(S(x_{j+1}-1,\ldots,x_{r}))|.  
\end{align*}

We now distinguish between two cases based on the value of $t$.
\begin{itemize}
    \item Case 1: $t\leq n-x_1$. In this case,  $\min\{x_1,n-t\}=x_1$. Thus, we can use Lemma $\ref{lm1}$ exactly $x_1$ times to obtain that
    \begin{align*}
\left|D_t(S(x_1,\ldots,x_r))\right|=|D_t(S(x_2,\ldots,x_r))|+\sum\limits_{i=0}^{x_1-1}
\sum\limits_{j=1}^{q_1-1}|D_{t-f(j)+i}(S(x_{j+1}-1,\ldots,x_{r}))|.  
\end{align*}

\item Case 2: $t > n-x_1$. Using Lemma $\ref{lm1}$ exactly $n-t$ times yields 
\begin{align*}
\left|D_t(S(x_1,\ldots,x_r))\right|=|D_t(S(x_1-(n-t),x_2,\ldots,x_r))|+\sum\limits_{i=0}^{n-t-1}
\sum\limits_{j=1}^{q_1-1}|D_{t-f(j)+i}(S(x_{j+1}-1,\ldots,x_{r}))|.  
\end{align*}
Since $|D_t(S(x_1-(n-t),x_2,\ldots,x_r))|=1$, and noticing that for
$t>|X|$ with $X\in \Sigma_q^*,$ $|D_t(X)|=0$, we have 
\begin{align*}
\left|D_t(S(x_1,\ldots,x_r))\right|=|D_t(S(x_2,\ldots,x_r))|+
\sum\limits_{i=0}^{x_1-1}\sum\limits_{j=1}^{q_1-1}|D_{t-f(j)+i}(S(x_{j+1}-1,\ldots,x_{r}))|+1|_{t>n-x_1}.
\end{align*}
\end{itemize}

So, the first result of the lemma follows.  

By a symmetric argument, the second identity regarding $x_r$ can be shown as $$\left|D_t(S(x_1,\ldots,x_r))\right|=|D_t(S(x_1,\ldots,x_{r-1}))|+
\sum\limits_{i=0}^{x_r-1}\sum\limits_{j=1}^{q_1-1}|D_{t-g(j)+i}(S(x_1,\ldots,x_{r-j}-1))|+1|_{t>n-x_r}.$$

This completes the proof.
\end{proof}

\begin{lemma}
\label{lm9}
For any string $U = S(x_1, \dots, x_r) \in \Sigma_q^*$ and $t \ge 1$, we have
\begin{align*}
\left|D_t(S(x_1, x_2, \ldots,x_{r-1}, x_r))\right|=\left|D_t(S(x_r,x_{r-1},\ldots,x_2,x_1))\right|.
\end{align*}
\end{lemma}
\begin{proof}
For convenience, let $U^T=S(x_r,x_{r-1},\ldots,x_2,x_1)$. we proceed by induction on $r$. 

\emph{Base Case}: When $r=1$ we have $U=0^{x_1}$ and $U^T=0^{x_1}$ such that $\left|D_t(U)\right|=\left|D_t(U^T)\right|$. Thus, the base of the induction holds.

\emph{Inductive Step}:  Assume the claim is true for cases where $r(U)\leq m$ with $m\geq 1$. Consider $r(U)=m+1$ and let $U=S(x_1,x_2,\ldots,x_{m+1})$. For convenience, let $n=\sum\limits_{i=1}^{m+1}x_{i}$ and define a function $f$ such that $f(i)=\sum\limits_{j=1}^{i}x_j$ for $i\in [q_1]$. Then $U^T=S(x_{m+1},\ldots,x_2,x_1)$. By Lemma $\ref{lm8}$, it follows that 
\begin{align}
\left|D_t(U)\right|=|D_t(S(x_2,\ldots,x_{m+1}))|+
\sum\limits_{i=0}^{x_1-1}\sum\limits_{j=1}^{q_1-1}|D_{t-f(j)+i}(S(x_{j+1}-1,\ldots,x_{m+1}))|+1|_{t>n-x_1}.\label{eq12}
\end{align}
Similarly, by Lemma $\ref{lm8}$, we have 
\begin{align}
\left|D_t(U^T)\right|=|D_t(S(x_{m+1},\ldots,x_2))|+
\sum\limits_{i=0}^{x_1-1}\sum\limits_{j=1}^{q_1-1}|D_{t-f(j)+i}(S(x_{m+1},\ldots,x_{j+1}-1))|+1|_{t>n-x_1}.\label{eq13}
\end{align}

Next we compare $\left|D_t(U)\right|$ and $\left|D_t(U^T)\right|$. By the induction  hypothesis, we have 
$$\left|D_{t}(S(x_2,\ldots,x_{m+1}))\right|=\left|D_{t}(S(x_{m+1},\ldots,x_{2}))\right|$$ and
$$|D_{t-f(j)+i}(S(x_{j+1}-1,\ldots,x_{m+1}))|=|D_{t-f(j)+i}(S(x_{m+1},\ldots,x_{j+1}-1))|,$$
for $i\in [x_1-1] \cup \{0\}$ and $j\in [\min\{q-1,m\}]$. Thus, combining with Eqs. $(\ref{eq12})$-$(\ref{eq13})$, it follows that $\left|D_t(U)\right|=\left|D_t(U^T)\right|.$
So, when $r(U)=m+1$, we have  $\left|D_t(U)\right|=\left|D_t(U^T)\right|$. 

By the induction hypothesis, the equality
$\left|D_t(U)\right|=\left|D_t(U^T)\right|$ holds for $r(U)\geq 1$. Therefore, the lemma follows.
\end{proof}

\begin{lemma}
\label{lm10}
For any string $S(x_1, \dots, x_r) \in \Sigma_q^*$ and $t \ge 1$, let $t_1 = \sum_{i=1}^{q_1} x_{r-i+1}$ and $t_2 = \sum_{i=1}^{q_1} x_i$. Then we have 
\begin{align*}
|D_t(S(x_1,\ldots,x_r))|=|D_t(S(x_1,\ldots,x_r-1))|+|D_{t-x_r}(S(x_1,\ldots,x_{r-1}))|-|D_{t-t_1}(S(x_1,\ldots,x_{r-q_1}-1))|.    
\end{align*}
Moreover, we have 
\begin{align*}
|D_t(S(x_1,\ldots,x_r))|=|D_t(S(x_1,\ldots,x_{r-1}))|+\sum\limits_{i=0}^{x_r-1}\big(|D_{t-x_r+i}(S(x_1,\ldots,x_{r-1}))|-|D_{t-t_1+i}(S(x_1,\ldots,x_{r-q_1}-1))|\big). 
\end{align*}
Similarly, we also have 
\begin{align*}
|D_t(S(x_1,\ldots,x_r))|=|D_t(S(x_2,\ldots,x_{r}))|+\sum\limits_{i=0}^{x_1-1}\big(|D_{t-x_1+i}(S(x_2,\ldots,x_{r}))|-|D_{t-t_2+i}(S(x_{q_1+1}-1,\ldots,x_{r}))|\big). 
\end{align*}
\end{lemma}
\begin{proof}
Let $a = r-1 \bmod {q_1}$. Then the element $a$ is the last element of $S(x_1,\ldots,x_r)$. By Lemma $\ref{lm1}$, we have 
\begin{align}
|D_t(S(x_1,\ldots,x_r))|=|D_t(S(x_1,\ldots,x_r))_{a}|+\left|\bigcup\limits_{i\in \Sigma_{q_1}\backslash\{a\}} D_t(S(x_1,\ldots,x_r))_{i}\right|. \label{eq:decom}   
\end{align}
Since $a$ is the terminal symbol of the sequence $S(x_1,\ldots,x_r)$, the first term in \eqref{eq:decom} simplifies to  
\begin{align*}
|D_t(S(x_1,\ldots,x_r))_{a}|=|D_t(S(x_1,\ldots,x_r-1))|.    
\end{align*}
The second term  $\bigcup\limits_{i\in \Sigma_{q_1}\backslash\{a\}}D_t(S(x_1,\ldots,x_r))_{i}$ represents the set of all subsequences of $S(x_1,\ldots,x_r)$ whose last element is not $a$. To obtain such a subsequence, the terminal run of $x_r$ symbols $a$ must be entirely deleted. Consequently, we have 
\begin{align*}
\bigcup\limits_{i\in \Sigma_{q_1}\backslash\{a\}} D_t(S(x_1,\ldots,x_r))_{i}\subseteq D_{t-x_r}(S(x_1,\ldots,x_{r-1})).
\end{align*}
Note that the set $D_{t-x_r}(S(x_1,\ldots,x_{r-1}))$ contains some subsequence of $S(x_1,\ldots,x_r)$ whose last element is $a$. 
By the definitions of $S(x_1,\ldots,x_r)$ and $t_1$, we have 
\begin{align*}
|D_{t-x_r}(S(x_1,\ldots,x_{r-1}))_a|=|D_{t-t_1}(S(x_1,\ldots,x_{r-q_1}-1))|.    
\end{align*}
Therefore, it follows that 
\begin{align*}
\left|\bigcup\limits_{i\in \Sigma_{q_1}\backslash\{a\}} D_t(S(x_1,\ldots,x_r))_{i}\right|&=|D_{t-x_r}(S(x_1,\ldots,x_{r-1}))|-|D_{t-x_r}(S(x_1,\ldots,x_{r-1}))_a|\\
&=|D_{t-x_r}(S(x_1,\ldots,x_{r-1}))|-|D_{t-t_1}(S(x_1,\ldots,x_{r-q_1}-1))|.
\end{align*}
Substituting these back into \eqref{eq:decom}, we obtain
\begin{align*}
|D_t(S(x_1,\ldots,x_r))|=|D_t(S(x_1,\ldots,x_r-1))|+|D_{t-x_r}(S(x_1,\ldots,x_{r-1}))|-|D_{t-t_1}(S(x_1,\ldots,x_{r-q_1}-1))|.    
\end{align*}
Furthermore, we decompose $D_t(S(x_1,\ldots,x_r-1))$ applying the above method $x_r-1$ times, it follows that 
\begin{align*}
|D_t(S(x_1,\ldots,x_r))|=|D_t(S(x_1,\ldots,x_{r-1}))|+\sum\limits_{i=0}^{x_r-1}\big(|D_{t-x_r+i}(S(x_1,\ldots,x_{r-1}))|-|D_{t-t_1+i}(S(x_1,\ldots,x_{r-q_1}-1))|\big).    
\end{align*}
A symmetric argument for the first run yields
\begin{align*}
|D_t(S(x_1,\ldots,x_r))|=|D_t(S(x_2,\ldots,x_{r}))|+\sum\limits_{i=0}^{x_1-1}\big(|D_{t-x_1+i}(S(x_2,\ldots,x_{r}))|-|D_{t-t_2+i}(S(x_{q_1+1}-1,\ldots,x_{r}))|\big). 
\end{align*}
Thus, the lemma follows.
\end{proof}
\subsection{Balancing operation}
 Informally, a string is said to be balanced if the lengths of its runs exhibit minimal variance. A balancing operation reduces this variance by monotonic transformations, that is shortening a long run while extending a short run. The following lemma characterizes conditions under which balancing increases the number of subsequences of a string. This result is essential for establishing the maximality of balanced strings.

\begin{definition}
A sequence of integers $x_{1} x_{2}\dots x_{n}$ is said to be \emph{symmetric} if it is identical to its reverse, i.e., $x_{k} = x_{n-k+1}$ for all $1 \le k \le n$.    
\end{definition} 
 
\begin{lemma}[Balancing Increases the Number of
Subsequences]
\label{lm11}
Let $t\geq 1$ and $U = S(x_1, \dots, x_r) \in \Sigma_q^*$. If $1 \le i < j \le r$ are indices such that $x_{i+1} \dots x_{j-1}$ is symmetric and $x_i - x_j > 1$,  then \begin{align*}
\left|D_t(S(x_1,\ldots,x_r))\right|\leq  
\left|D_t(S(x_1,\ldots,x_{i-1},x_i-1,x_{i+1}\ldots, x_{j-1},x_j+1,x_{j+1},\ldots,x_r))\right|.
\end{align*}   
\end{lemma}

To prove Lemma \ref{lm11}, we require the following lemma, which generalizes  results for binary strings from \cite[Lemma 8]{Liron} to $q$-ary strings.

\begin{lemma}
\label{lm12}   
 Let $t\geq 1$ and $x_1>x_{r}\geq 1$. If the sequence $x_{2}\ldots x_{r-1}$ is symmetric, we have the following results:
\begin{enumerate}
    \item $\left|D_t(S(x_1,\ldots,x_r))\right|\leq \left|D_t(S(x_1-1,x_2,\ldots,x_{r-1},x_{r}+1))\right|$.
    \item If $z\geq 1$ then $\left|D_t(S(x_1,\ldots,x_r,z))\right|\leq \left|D_t(S(x_1-1,x_2,\ldots,x_{r-1},x_{r}+1,z))\right|$.
    \item If $y\geq 1$ and $x_1-x_r\geq 2$ then $\left|D_t(S(y,x_1,\ldots,x_r))\right|\leq \left|D_t(S(y,x_1-1,x_2,\ldots,x_{r-1},x_{r}+1))\right|$.
    \item If $y,z\geq 1$ and $x_1-x_r\geq 2$ then $\left|D_t(S(y,x_1,\ldots,x_r,z))\right|\leq \left|D_t(S(y,x_1-1,x_2,\ldots,x_{r-1},x_{r}+1,z))\right|$.
\end{enumerate}
\end{lemma}

\begin{proof} 
1) For convenience, let $U=S(x_1,...,x_r)$ and $V=S(x_1-1,\ldots,x_{r-1},x_r+1)$. We further define $V^T=S(x_r+1,x_{r-1},\ldots,x_{2},x_1-1)$. Note that $V^T=S(x_r+1,x_{2},\ldots,x_{r-1},x_1-1)$ since  $x_2\dots x_{r-1}$ is symmetric. By Lemma $\ref{lm9}$ we have $|D_t(V)|=|D_t(V^T)|$. By Lemma $\ref{lm1}$, it follows that 
\begin{align}
\left|D_t(U)\right|=&\sum\limits_{i=0}^{q_1-1}\left|D_t(U)^{i}\right|=\left|D_t(U)^{0}\right|+\sum\limits_{i=1}^{q_1-1}\left|D_t(U)^{i}\right|\nonumber\\ 
=&\left|D_t\big(S(x_1-1, x_2,\ldots,x_{r-1}, x_{r})\big)\right|+\sum\limits_{i=1}^{q_1-1}\left|D_{t-x_1}(S(x_2,\ldots,x_{r};1,\ldots,(r-1)~\bmod q_1))^{i}\right|,\label{eq14}
\end{align}
  Similarly, applying Lemma $\ref{lm1}$ to  $V^T$ yields
\begin{align}
|&D_t(V^T)|=\sum\limits_{i=0}^{q_1-1}\left|D_t(V^T)^{i}\right|=\left|D_t(V^T)^0\right|+\sum\limits_{i=1}^{q_1-1}\left|D_t(V^T)^{i}\right|\nonumber\\ 
=&\left|D_t\big(S(x_r,x_{r-1},\ldots,x_{2},x_{1}-1)\big)\right|+\sum\limits_{i=1}^{q_1-1}\left|D_{t-x_{r}-1}(S(x_{r-1},\ldots,x_{2},x_1-1;1,\ldots,(r-1) \bmod q_1))^{i}\right|\nonumber\\
\overset{(a)}{=}&\left|D_t\big(S(x_1-1,x_2,\ldots,x_{r-1},x_{r})\big)\right|+\sum\limits_{i=1}^{q_1-1}\left|D_{t-x_{r}-1}(S(x_2,\ldots,x_{r-1},x_1-1;1,\ldots,(r-1) \bmod q_1))^{i}\right|,\label{eq14~~}
\end{align}
where Equality $(a)$ follows from Lemma $\ref{lm9}$ and the symmetry of $x_2\ldots x_{r-1}$. 

Under the assumption $x_1>x_r$, it is observed that
\begin{align*}
S(x_2,\ldots,x_r;1,\ldots,(r-1)~\bmod q_1)\in D_{x_1-x_r-1}(S(x_2,\ldots,x_{r-1},x_1-1;1,\ldots,(r-1)~\bmod q_1)).
\end{align*}
Thus, by Lemma $\ref{lm4}$, it follows that for each $i\in [q_1-1]$,
\begin{align*}
\left|D_{t-x_1}(S(x_2,\ldots,x_{r};1,\ldots,(r-1) \bmod q_1))^{i}\right|\leq \left|D_{t-x_{r}-1}(S(x_2,\ldots,x_{r-1},x_1-1;1,\ldots,(r-1) \bmod q_1))^{i}\right|.
\end{align*}
 Comparing the corresponding terms in \eqref{eq14} and \eqref{eq14~~}, we conclude that $|D_t(U)|\leq |D_t(V^T)|.$ Since $|D_t(V)|=|D_t(V^T)|$, the inequality $|D_t(U)|\leq |D_t(V)|$ holds. 

2) For convenience, let $U_1=S(x_1,...,x_r,z)$ and  $V_1=S(x_1-1,\ldots,x_{r-1},x_r+1,z)$. We define $t_1=z+\sum\limits_{i=1}^{q_1-1}x_{r-i+1}$ and $t_2=z+1+\sum\limits_{i=1}^{q_1-1}x_{r-i+1}=t_1+1$. By Lemma $\ref{lm10}$, we have 
\begin{align*}
|D_t(U_1)|&=|D_t(S(x_1,\ldots,x_r,z))|\\&=|D_t(S(x_1,\ldots,x_r))|+\sum\limits_{i=0}^{z-1}\big(|D_{t-z+i}(S(x_1,\ldots,x_{r}))|-|D_{t-t_1+i}(S(x_1,\ldots,x_{r+1-q_1}-1))|\big),\\    
|D_t(V_1)|&=|D_t(S(x_1-1,\ldots,x_r+1,z))|\\&=|D_t(S(x_1-1,\ldots,x_r+1))|+\sum\limits_{i=0}^{z-1}\big(|D_{t-z+i}(S(x_1-1,\ldots,x_{r}+1))|-|D_{t-t_2+i}(S(x_1-1,\ldots,x_{r+1-q_1}-1))|\big). 
\end{align*}

Since $x_1>x_r$ and $x_2\ldots x_{r-1}$ is symmetric,  Case $1)$ of Lemma $\ref{lm12}$ implies that 
\begin{align*}
|D_t(S(x_1,\ldots,x_r))|&\leq |D_t(S(x_1-1,\ldots,x_r+1))|,\\ 
|D_{t-z+i}(S(x_1,\ldots,x_r))|&\leq |D_{t-z+i}(S(x_1-1,\ldots,x_r+1))|,
\end{align*}
for $i\in [z-1]\cup\{0\}$.
Regarding the subtracted terms in the summations, we observe that $t_2=t_1+1$ and $S(x_1-1,\ldots,x_{r+1-q_1}-1)\in D_1(S(x_1,\ldots,x_{r+1-q_1}-1))$. By Lemma $\ref{lm4}$, it follows that 
\begin{align*}
|D_{t-t_2+i}(S(x_1-1,\ldots,x_{r+1-q_1}-1))|\leq |D_{t-t_1+i}(S(x_1,\ldots,x_{r+1-q_1}-1))|   
\end{align*}
for $i\in [z-1]\cup\{0\}$.
Comparing the corresponding terms in the expressions for $|D_t(U_1)|$ and $|D_t(V_1)|$, we have  $|D_t(U_1)|\leq |D_t(V_1)|$.

3) For convenience, let $U_2=S(y,x_1,...,x_r)$ and $V_2=S(y,x_1-1,\ldots,x_{r-1},x_r+1)$. We define  $n=y+\sum\limits_{i=1}^{r}x_i$, and let $f(j)=y+\sum\limits_{i=1}^{j-1}x_i$ for $j\in [q_1-1]$. By Lemma $\ref{lm8}$, we have 
\begin{align*}
|D_t(U_2)|=|D_t(S(y,x_1,\ldots,x_r))|&=|D_t(S(x_1,\ldots,x_r))|+\sum\limits_{i=0}^{y-1}|D_{t-y+i}(S(x_{1}-1,\ldots,x_{r}))|\\
&~~~~+\sum\limits_{i=0}^{y-1}\sum\limits_{j=2}^{q_1-1}|D_{t-f(j)+i}(S(x_{j}-1,\ldots,x_{r}))|+1|_{t>n-y},\\    
|D_t(V_2)|=|D_t(S(y,x_1-1,\ldots,x_r+1))|&=|D_t(S(x_1-1,\ldots,x_r+1))|+\sum\limits_{i=0}^{y-1}|D_{t-y+i}(S(x_{1}-2,\ldots,x_{r}+1))|\\
&~~~~+\sum\limits_{i=0}^{y-1}\sum\limits_{j=2}^{q_1-1}|D_{t-f(j)+i+1}(S(x_{j}-1,\ldots,x_{r}+1))|+1|_{t>n-y}. 
\end{align*}

Given $x_1-x_r\geq 2$, we have $x_1>x_r$ and $x_1-1>x_r$. Since $x_2\ldots x_{r-1}$ is symmetric, Case $1)$ of Lemma $\ref{lm12}$ yields
\begin{align*}
|D_t(S(x_1,\ldots,x_r))|&\leq |D_t(S(x_1-1,\ldots,x_r+1))|,\\ 
|D_{t-y+i}(S(x_1-1,\ldots,x_r))|&\leq |D_{t-y+i}(S(x_1-2,\ldots,x_r+1))|,
\end{align*}
for $i\in [y-1]\cup \{0\}$.
Furthermore, we address the double summation terms. Since $S(x_j-1,\ldots,x_{r})\in D_1(S(x_j-1,\ldots,x_{r}+1))$, by Lemma $\ref{lm4}$, it follows that 
\begin{align*}
|D_{t-f(j)+i}(S(x_j-1,\ldots,x_{r}))|\leq |D_{t-f(j)+i+1}(S(x_j-1,\ldots,x_{r}+1))|   
\end{align*}
for $i\in [y-1]\cup\{0\}$ and $j\in [q_1-1]\backslash\{1\}$.
Therefore, we have  $|D_t(U_2)|\leq |D_t(V_2)|$.

4) For convenience, let $U_3=S(y,x_1,...,x_r,z)$ and $V_3=S(y,x_1-1,\ldots,x_{r-1},x_r+1,z)$. 
We let  $n=z+y+\sum\limits_{i=1}^{r}x_i$, and define $f(j)=y+\sum\limits_{i=1}^{j-1}x_i$ for $j\in [q_1-1]$. By Lemma $\ref{lm8}$, we have 
\begin{align*}
|D_t(U_3)|=|D_t(S(y,x_1,\ldots,x_r,z))|&=|D_t(S(x_1,\ldots,x_r,z))|+\sum\limits_{i=0}^{y-1}|D_{t-y+i}(S(x_{1}-1,\ldots,x_{r},z))|\\
&~~~~+\sum\limits_{i=0}^{y-1}\sum\limits_{j=2}^{q_1-1}|D_{t-f(j)+i}(S(x_{j}-1,\ldots,x_{r},z))|+1|_{t>n-y},\\    
|D_t(V_3)|=|D_t(S(y,x_1-1,\ldots,x_r+1,z))|&=|D_t(S(x_1-1,\ldots,x_r+1,z))|+\sum\limits_{i=0}^{y-1}|D_{t-y+i}(S(x_{1}-2,\ldots,x_{r}+1,z))|\\
&~~~~+\sum\limits_{i=0}^{y-1}\sum\limits_{j=2}^{q_1-1}|D_{t-f(j)+i+1}(S(x_{j}-1,\ldots,x_{r}+1,z))|+1|_{t>n-y}. 
\end{align*}
Since $x_1-x_r\geq 2$, we have $x_1>x_r$ and $x_1-1>x_r$. When $x_2\ldots x_{r-1}$ is symmetric, by the result of the Case $2)$ in Lemma $\ref{lm12}$, we have 
\begin{align*}
|D_t(S(x_1,\ldots,x_r,z))|&\leq |D_t(S(x_1-1,\ldots,x_r+1,z))|,\\ 
|D_{t-y+i}(S(x_1-1,\ldots,x_r,z))|&\leq |D_{t-y+i}(S(x_1-2,\ldots,x_r+1,z))|,
\end{align*}
for $i\in [y-1]\cup \{0\}$.
Moreover, since $S(x_j-1,\ldots,x_{r},z)\in D_1(S(x_j-1,\ldots,x_{r}+1,z))$, by Lemma $\ref{lm4}$, it follows that 
\begin{align*}
|D_{t-f(j)+i}(S(x_j-1,\ldots,x_{r},z))|\leq |D_{t-f(j)+i+1}(S(x_j-1,\ldots,x_{r}+1,z))|   
\end{align*}
for $i\in [y-1]\cup\{0\}$ and $j\in [q_1-1]\backslash\{1\}$.
Therefore, we have  $|D_t(U_3)|\leq |D_t(V_3)|$.
\end{proof}

Now we will prove Lemma $\ref{lm11}$.
\begin{proof}  We prove Lemma $\ref{lm11}$ by induction on $p$, the number of runs in $U$ outside of the subsequence $S(x_i,\ldots,x_{j};(i-1) \bmod q_1,\ldots,(j-1) \bmod q_1)$, where $p=(i-1)+(r-j)=r+i-j-1$. For convenience, we refer to these as outer runs. 

\emph{Base Case}: When $p=1$,  there is exactly one outer run. The lemma is reduced to Lemma $\ref{lm12}~2)$ or $\ref{lm12}~3)$ which indicates that the base of the induction holds. 

\emph{Induction Step}: Assume the claim holds for all cases where the number of outer runs is at most $p$ with $p \ge 1$. Consider the case with $p+1$ outer runs, that is, $r+i-j-1=p+1$. Note that there are at least two outer runs. If there is exactly one outer run on each side (i.e., $i=2$ and $j=r-1$), this is the case of Lemma $\ref{lm12}~ 4)$. Otherwise, there must be at least two runs on one side, that is, $i>2$ or $j<r-1$. We analyze these two scenarios separately.

Let $U=S(x_1,\ldots,x_{i-1},x_i,x_{i+1},\ldots,x_{j-1},x_{j},x_{j+1},\ldots,x_r)$ and  $V=S(x_1,\ldots,x_{i-1},x_i-1,x_{i+1},\ldots,x_{j-1},x_{j}+1,x_{j+1},$\\$\ldots,x_r)$. The total length of both sequences is denoted by  $n = \sum\limits_{w=1}^{r}x_w$.

\textbf{Case 1: $i>2$}. Define $f(l)=\sum\limits_{w=1}^{l}x_w$ for $l\in [q_1-1]$. 
By Lemma $\ref{lm8}$, we have
\begin{small}
\begin{align*}
|D_t(U)|&=|D_t(S(x_2,\ldots,x_{i-1},x_{i},\ldots,x_j,\ldots,x_r))|+
\sum\limits_{k=0}^{x_1-1}\sum\limits_{l=1}^{\min\{i-2,q_1-1\}}|D_{t-f(l)+k}(S(x_{l+1}-1,\ldots,x_{i},\ldots,x_j,\ldots,x_{r}))|\\
&~~~~+\sum\limits_{k=0}^{x_1-1}\sum\limits_{l=i-1,l\leq q_1-1}|D_{t-f(l)+k}(S(x_{i}-1,x_{i+1},\ldots,x_{j},\ldots,x_{r}))|\\
&~~~~+\sum\limits_{k=0}^{x_1-1}\sum\limits_{l=i,l\leq q_1-1}^{j-2}|D_{t-f(l)+k}(S(x_{l+1}-1,\ldots,x_{j},\ldots,x_{r}))|\\
&~~~~+\sum\limits_{k=0}^{x_1-1} \sum\limits_{l=j-1,l\leq q_1-1} |D_{t-f(l)+k}(S(x_{j}-1,\ldots,x_{r}))|\\
&~~~~+\sum\limits_{k=0}^{x_1-1}\sum\limits_{l=j}^{q_1-1}|D_{t-f(l)+k}(S(x_{l+1}-1,\ldots,x_{r}))|+1|_{t>n-x_1},\\
|D_t(V)|&=|D_t(S(x_2,\ldots,x_{i-1},x_{i}-1,\ldots,x_j+1,\ldots,x_r))|+
\sum\limits_{k=0}^{x_1-1}\sum\limits_{l=1}^{\min\{i-2,q_1-1\}}|D_{t-f(l)+k}(S(x_{l+1}-1,\ldots,x_{i}-1,\ldots,x_j+1,\ldots,x_{r}))|\\
&~~~~+\sum\limits_{k=0}^{x_1-1}\sum\limits_{l=i-1,l\leq q_1-1}|D_{t-f(l)+k}(S(x_{i}-2,x_{i+1},\ldots,x_{j}+1,\ldots,x_{r}))|\\
&~~~~+\sum\limits_{k=0}^{x_1-1}\sum\limits_{l=i,l\leq q_1-1}^{j-2}|D_{t-f(l)+1+k}(S(x_{l+1}-1,\ldots,x_{j}+1,\ldots,x_{r}))|\\
&~~~~+\sum\limits_{k=0}^{x_1-1}\sum\limits_{l=j-1, l\leq q_1-1}|D_{t-f(l)+1+k}(S(x_{j}, \dots, x_r))|\\
&~~~~+\sum\limits_{k=0}^{x_1-1}\sum\limits_{l=j}^{q_1-1}|D_{t-f(l)+k}(S(x_{l+1}-1,\ldots,x_{r}))|+1|_{t>n-x_1}.\\
\end{align*}
\end{small}
The two expressions in the first two lines are comparable term by term using the induction hypothesis, as in each argument the number of outer runs is decreased at least by one. In the third line, by Lemma $\ref{lm4}$, we easily obtain that 
\begin{align*}
|D_{t-f(l)+k}(S(x_{l+1}-1,\ldots,x_{j},\ldots,x_{r}))|\leq |D_{t-f(l)+1+k}(S(x_{l+1}-1,\ldots,x_{j}+1,\ldots,x_{r}))|,    
\end{align*}
for $k\in [x_1-1]\cup \{0\}$ and $l\in [i,j-2]$. In the forth line, it similarly follows from Lemma $\ref{lm4}$ that
\[
|D_{t-f(l)+k}(S(x_j-1,\dots,x_r))| \leq |D_{t-f(l)+1+k}(S(x_j,\dots,x_r))|,
\]
for $k \in [x_1-1] \cup \{0\}$ and $l=j-1$.
The two expressions in the  fifth line are identical.

Therefore, we have $|D_t(U)|\leq |D_t(V)|$ for $r+i-j-1=p+1$ and $i>2$. 

\textbf{Case 2:} $j<r-1$. 

When $r-q_1 < j$, set
$t_1=\sum\limits_{l=1}^{q_1}x_{r-l+1}$. Then by Lemma $\ref{lm10}$, we have 
\begin{small}
\begin{align*}
|D_t(U)|&=|D_t(S(x_1,\ldots,x_i,\ldots,x_j,\ldots,x_r))|\\
&=|D_t(x_1,\ldots,x_i,\ldots,x_j,\ldots,x_{r-1})|+\sum\limits_{k=0}^{x_r-1}\big(|D_{t-x_r+k}(S(x_1,\ldots,x_i,\ldots,x_j,\ldots,x_{r-1}))|-|D_{t-t_1+k}(S(x_1,\ldots,x_{r-q_1}-1))|\big),\\    
|D_t(V)|&=|D_t(S(x_1,\ldots,x_i-1,\ldots,x_j+1,\ldots,x_r))|\\
&=|D_t(x_1,\ldots,x_i-1,\ldots,x_j+1,\ldots,x_{r-1})|+\sum\limits_{k=0}^{x_r-1}\big(|D_{t-x_r+k}(S(x_1,\ldots,x_i-1,\ldots,x_j+1,\ldots,x_{r-1}))|\\
&~~~~-|D_{t-t_2+k}(S'(x_1,\ldots,x_{r-q_1}-1))|\big),
\end{align*}
where 
\begin{equation*}
S'(x_1,\ldots,x_{r-q_1}-1))=
\begin{cases}
S(x_1,\ldots,x_i-1,\ldots,x_{r-q_1}-1)&\text{when $i+1\leq r-q_1<j$}\\
S(x_1,\ldots,x_i-2)&\text{when $r-q_1=i$}\\
S(x_1,\ldots,x_{r-q_1}-1)&\text{when $r-q_1<i$},
\end{cases}
\end{equation*}
and 
\begin{equation*}
t_2=
\begin{cases}
t_1+1&\text{when $i+1\leq r-q_1<j$}\\
t_1+1&\text{when $r-q_1=i$}\\
t_1&\text{when $r-q_1<i$}.
\end{cases}
\end{equation*}
\end{small}
Since $x_i-x_j>1$ and $x_{i+1}\ldots x_{j-1}$ is symmetric, by the assumption, we have 
\begin{align*}
|D_t(x_1,\ldots,x_i,\ldots,x_j,\ldots,x_{r-1})|&\leq |D_t(x_1,\ldots,x_i-1,\ldots,x_j+1,\ldots,x_{r-1})|,\\ 
|D_{t-x_r+k}(x_1,\ldots,x_i,\ldots,x_j,\ldots,x_{r-1})|&\leq |D_{t-x_r+k}(x_1,\ldots,x_i-1,\ldots,x_j+1,\ldots,x_{r-1})|, 
\end{align*}
for $k\in [x_r-1]\cup\{0\}$.
Moreover, it remains to compare the values  $|D_{t-t_1+k}(S(x_1, \dots, x_{r-q_1}-1))|$ and $|D_{t-t_2+k}(S'(x_1, \dots,\\ x_{r-q_1}-1))|$ for $k \in [x_r-1] \cup \{0\}$. 
For $r-q_1 <j$, we consider the following three subcases.
\begin{itemize}
    \item When $i+1\leq r-q_1<j$, we have $t_2=t_1+1$ and $S(x_1,\ldots,x_{i}-1,\ldots,x_{r-q_1}-1)\in D_1(S(x_1,\ldots,x_i,\ldots,x_{r-q_1}-1))$, by Lemma $\ref{lm4}$, it follows that 
\begin{align*}
|D_{t-t_2+k}(S(x_1,\ldots,x_i-1,\ldots,x_{r-q_1}-1))|\leq |D_{t-t_1+k}(S(x_1,\ldots,x_i,\ldots,x_{r-q_1}-1))|   
\end{align*}
for $k\in [x_r-1]\cup\{0\}$.
\item When $r-q_1=i$, we have $t_2=t_1+1$ and $S(x_1,\dots,x_i-2) \in D_1(S(x_1,\ldots,x_{r-q_1}-1))$, by Lemma $\ref{lm4}$, it follows that
\begin{align*}
    |D_{t-t_2+k}(S(x_1,\dots,x_i-2))| \leq |D_{t-t_1+k}(S(x_1,\ldots,x_{r-q_1}-1))|
\end{align*}
for $k \in [x_r-1] \cup \{0\}$.
\item When $r-q_1<i$, the two expressions are identical for $k \in [x_r-1]\cup\{0\}$. 
\end{itemize} 
 Therefore, when  $r-q_1<j$, we have  $|D_t(U)|\leq |D_t(V)|$.

When $r-q_1\geq j$, define $g(l)=\sum\limits_{w=1}^{l}x_{r-w+1}$ for $l\in [q_1-1]$. 
Since $r-q_1+1\geq j+1$, it follows from Lemma $\ref{lm8}$ that
\begin{small}
\begin{align*}
|D_t(U)|&=|D_t(S(x_1,\ldots,x_{i-1},x_{i},\ldots,x_j,\ldots,x_{r-1}))|+
\sum\limits_{k=0}^{x_r-1}\sum\limits_{l=1}^{q_1-1}|D_{t-g(l)+k}(S(x_{1},\ldots,x_{i},\ldots,x_j,\ldots,x_{r-l}-1))|\\&~~~~+1_{t>n-x_r},\\
|D_t(V)|&=|D_t(S(x_1,\ldots,x_{i-1},x_{i}-1,\ldots,x_j+1,\ldots,x_{r-1}))|+
\sum\limits_{k=0}^{x_r-1}\sum\limits_{l=1}^{q_1-1}|D_{t-g(l)+k}(S(x_1,\ldots,x_{i}-1,\ldots,x_j+1,\ldots,x_{r-l}-1))|\\&~~~~+1_{t>n-x_r}.
\end{align*}
\end{small} The two expressions in the first line are comparable term by term using the induction hypothesis, as in each argument the number of outer runs is decreased at least by one. Therefore, when $r-q_1\geq j$ we have $|D_t(U)|\leq |D_t(V)|$. 

By the above discussion, it follows that $|D_t(U)|\leq |D_t(V)|$ for $r+i-j-1=p+1$ and $j<r-1$.

Therefore, the claim holds for $r+i-j-1=p+1$. So, the lemma follows.
\end{proof}

\subsection{Balanced strings}
By Lemma $\ref{lm5}$, $|D_t(U)|$ is identical for all strings $U \in \mathcal{B}_{r,k}^q$. Thus, for $\mathcal{B}_{r,k}^q$, it suffices to consider the balanced string $S(r\times k)$. For notation brevity, we use $B_{r,k;q}$ to represent this $q$-ary string. In the following, we demonstrate that among all $q$-ary strings of length $rk$ with $r$ runs, $B_{r,k;q}$ achieves the maximum number of subsequences.  

\begin{theorem}
\label{thm2}
Let $U=S(x_1,\ldots,x_r;a_1,\ldots,a_r)$, $n=\sum\limits_{i=1}^{r}x_i$, $a_i\in \Sigma_q$ for $i\in [r]$, and $k=n/r$. If $k$ is an integer then $|D_t(U)|\leq |D_t(B_{r,k;q})|$.
\end{theorem}
\begin{proof}
We construct a finite sequence of strings $U_0, U_1, \dots, U_m$ such that $U_0 = U$, $U_m = B_{r,k;q}$, and $|D_t(U_i)| \le |D_t(U_{i+1})|$ for all $0 \le i < m$.

First, performing the cyclic operation on $U_0$ yields $U_1=S(x_1, \dots, x_r)$. By Lemma \ref{lm6} and Remark \ref{rmk1}, it follows that $|D_t(U_0)| \le |D_t(U_1)|$.

Next, $U_1$ is transformed into $B_{r,k;q}$ through repeated balancing operations. For any string $U_i\neq B_{r,k;q}$, let $U_i=S(x_1^{(i)},\ldots,x_r^{(i)})$. Since $U_i$ is unbalanced, there exists at least one run length different from $k$ (w.l.o.g., $x_j^{(i)} > k$), which implies the existence of another run length smaller than $k$. We then select a pair of indices $(p, s)$ with $p < s$ such that $|x_p^{(i)} - x_s^{(i)}| > 1$, where $s - p$ is chosen to be minimal. Without loss of generality, assume $x_p^{(i)} > x_s^{(i)}$. The minimality of $s-p$ implies $x_p^{(i)}>x_{p+1}^{(i)}=\ldots=x_{s-1}^{(i)}>x_s^{(i)}$. We define $U_{i+1}$ by decreasing $x_p^{(i)}$ by 1 and increasing $x_s^{(i)}$ by 1. Since $U_i$ and $U_{i+1}$ satisfy the conditions of Lemma \ref{lm7}, we have $|D_t(U_i)| \le |D_t(U_{i+1})|$. This process is finite since the sum $\sum_{j=1}^{r}\left(x_j^{(i)}\right)^2$ is a positive integer that strictly decreases at each step.

Thus, we obtain the non-decreasing chain:
\begin{equation*}
|D_t(U)| \le |D_t(U_1)| \le \dots \le |D_t(B_{r,k;q})|.
\end{equation*}
The theorem follows.
\end{proof}

We now illustrate the balancing process of Theorem $\ref{thm2}$ by transforming an unbalanced string $U_0$ to $B_{r,k;q}$ in the following example.
\begin{example}
\label{ex1}
Let $r=6$, $k=4$, $q=4$, and $t=7$. Given an unbalanced string $U_0=S(7,2,1,3,5,6;0,1,0,2,0,3)$, we perform the cyclic operation on $U_0$ to obtain $U_1=S(7,2,1,3,5,6)$. Next, we apply a series of balancing operations to $U_1$ to obtain the balanced string $U_{10}=S(4,4,4,4,4,4)$, with the specific steps detailed in Table \ref{tab1}.
\begin{table}[htbp]
  \centering
  \caption{A process of transforming $U_0$ to $U_{10}$}
  \label{tab1}
  \begin{tabular}{lcccr}
    \toprule
    $i$ & $U_i$ & runs & $\sum\limits_{j=1}^{6}x_j^2$& $|D_7(U_i)|$ \\
    \midrule
    $0$ & $000000011022200000333333$ & $7,2,1,3,5,6$ & $124$ & $326$\\
    $1$ & $000000011233300000111111$ & $7,2,1,3,5,6$ & $124$ & $378$\\
    $2$ & $000000011233330000111111$ & $7,2,1,4,4,6$ & $122$ & $394$\\
     $3$ & $000000111233330000111111$& $6,3,1,4,4,6$& $114$ & $448$\\
    $4$ & $000001111233330000111111$& $5,4,1,4,4,6$& $110$ & $473$\\
    $5$ & $000001112233330000111111$& $5,3,2,4,4,6$& $106$& $562$\\
    $6$ & $000011112233330000111111$& $4,4,2,4,4,6$& $104$& $581$\\
    $7$ & $000011122233330000111111$& $4,3,3,4,4,6$& $102$& $616$ \\
    $8$ & $000011122233330000011111$ & $4,3,3,4,5,5$& $100$& $626$\\
    $9$ & $000011122223333000011111$ & $4,3,4,4,4,5$& $98$& $646$\\
    $10$& $000011112222333300001111$ & $4,4,4,4,4,4$ & $96$&$666$\\
    \bottomrule\end{tabular}
\end{table}
\end{example}

When $r\nmid n$, the following corollary provides an upper bound on $|D_t(S(x_1,\ldots,x_r;a_1,\dots, a_r))|$, where $n=\sum_{i=1}^{r}x_i$.
\begin{corollary}
\label{cor1}
Let $U=S(x_1,\ldots,x_r;a_1,\ldots,a_r)$ with $ n=\sum_{i=1}^{r}x_i, a_i\in \Sigma_q$ for $i\in [r]$, and $k=\lceil n/r \rceil$, we have 
$|D_t(U)|\leq |D_t(B_{r,k;q})|$.
\end{corollary}
\begin{proof}
When $r\nmid n$, we let $s=rk-n$ and $V=S(x_1,\ldots,x_{r-1},x_r+s;a_1,\ldots,a_r)$. By Lemma $\ref{lm3}$, we have $|D_t(U)|\leq |D_t(V)|$. Since $|V|=rk$ and $r(V)=r$, Theorem $\ref{thm2}$ implies $|D_t(V)|\leq |D_t(B_{r,k;q})|$.     
\end{proof}

\begin{remark}
\label{rmk3}
Theorem  $\ref{thm2}$ and Corollary $\ref{cor1}$ extend the  results of \cite{Liron} from the binary case to the general $q$-ary setting. For $r \nmid n$, Corollary $\ref{cor1}$ provides a universal upper bound that is consistently tighter than the results in \cite{L0} and \cite{Hirschberg}, as demonstrated by the numerical comparisons in Table \ref{tab:comparison}.
\end{remark}

\begin{table}[htbp]
\centering
\caption{Comparison of Upper Bounds for $r \nmid n$}
\label{tab:comparison}
\begin{tabular}{lccccc}
\toprule
$(n, r, q, t)$ & \cite{L0} & \cite{Hirschberg}  & \textbf{Our Bound (Cor.~\ref{cor1})} \\
\midrule
$(13, 3, 3, 10)$ & $66$ & $27$ & $\mathbf{21}$  \\
$(17, 5, 3, 5)$ & $126$ & $5096$ &  $\mathbf{121}$  \\
$(26, 5, 4, 8)$ & $495$ & $1430758$ & $\mathbf{470}$  \\
\bottomrule
\end{tabular}
\end{table}

\section{Our upper bound}\label{sec5}

In this section, we establish an upper bound on the number of distinct subsequences obtained by deleting symbols from a string. Specifically, we give a recursive formula for the exact number of subsequences of a balanced string in $\Sigma_q^*$. We then find its closed-form expression, and use it to obtain a tight upper bound on the number of subsequences for general strings. In addition, we show the improvement that our upper bound provides.

\subsection{Recursive expression}
\begin{definition}
\label{def1}  
For any $r,k,q$, let $B_{r,k;q}'$ be the string derived from $B_{r,k;q}$ by deleting the first symbol. E.g. $B_{4,2;3}'=S(1,2,2,2)=0112200$.
\end{definition}

For convenience, let $b(r,k,t;q)=|D_t(B_{r,k;q})|$ and $b'(r,k,t;q)=|D_t(B_{r,k;q}')|$.

\begin{lemma}
\label{lm13}
Let $r,k,q,t$ be positive integers. If $q\geq r$, then
\begin{align*}
b(r,k,t;q)=\sum\limits_{i=0}^{k}b(r-1,k,t-k+i;q).    
\end{align*}

If $q<r$, then
 \begin{align*}
b(r,k,t;q)=b'(r,k,t;q)+\sum\limits_{i=1}^{q-1} b'(r-i,k,t-ik;q).  
\end{align*}
\end{lemma}
\begin{proof}
When $q\geq r$, Lemma $\ref{lm10}$ directly implies the first result. When $q<r$, the second result is obtained directly from Lemma $\ref{lm1}$.  
\end{proof}

When $k$ and $q$ are known from the context, we will use the short notations $b(r,t)$ and $b'(r,t)$ for $b(r,k,t;q)$ and $b'(r,k,t;q)$, respectively. The following recurrence is essential for deriving the closed-form formula of $b(r, t)$.

\begin{lemma}[Recursive Expression for $b(r,t)$ and $b'(r,t)$]
\label{lm14}
When $q\geq r$, we have 
\begin{equation*}
b(r,t)=
\begin{cases}
0, &\text{if $t<0$ or $t>kr$;}~\\
\sum\limits_{i=0}^{t}b(r-1,i), &\text{if $0\leq t\leq k$;}~\\
\sum\limits_{i=0}^{k}b(r-1,t-i), &\text{otherwise.}
\end{cases}
\end{equation*}
Moreover, we have 
\begin{equation*}
b'(r,t)=
\begin{cases}
(A)~~0, &\text{if $t<0$ or $t\geq kr$;}~\\
(B)~~1+\sum\limits_{i=1}^{k-1}\sum\limits_{j=1}^{q-1}b'(r-j,t-jk+i), &\text{if $k(r-1)\leq t\leq kr-1$;}~\\
(C)~~\sum\limits_{j=0}^{q-1}b'(r-1-j,t-jk)+\sum\limits_{i=1}^{k-1}\sum\limits_{j=1}^{q-1}b'(r-j,t-jk+i), &\text{otherwise.}
\end{cases}
\end{equation*}

\end{lemma}
\begin{proof}
When $q\geq r$, using Lemma $\ref{lm13}$, we get the recursive expression for $b(r,t)$. 

Using Lemma $\ref{lm8}$, we have 
\begin{align*}
b'(r,t)=b(r-1,t)+\sum\limits_{i=1}^{k-1}\sum\limits_{j=1}^{q-1}b'(r-j,t-jk+i)+1|_{t>(r-1)k}.    
\end{align*}
Moreover,  using Lemma $\ref{lm1}$ we have 
\begin{align*}
b(r-1,t)=b'(r-1,t)+\sum\limits_{j=1}^{q-1}b'(r-1-j,t-jk).    
\end{align*}
If $t<k(r-1)$ then we have 
\begin{align*}
b'(r,t)=\sum\limits_{j=0}^{q-1}b'(r-1-j,t-jk)+\sum\limits_{i=1}^{k-1}\sum\limits_{j=1}^{q-1}b'(r-j,t-jk+i).    
\end{align*}
If $t=k(r-1)$ then $b(r-1,t)=1$ and $1|_{t>(r-1)k}=0$, and thus 
\begin{align*}
b'(r,t)=1+\sum\limits_{i=1}^{k-1}\sum\limits_{j=1}^{q-1}b'(r-j,t-jk+i).    
\end{align*}
If $t>k(r-1)$ then $b(r-1,t)=0$ and $1|_{t>(r-1)k}=1$, and thus 
\begin{align*}
b'(r,t)=1+\sum\limits_{i=1}^{k-1}\sum\limits_{j=1}^{q-1}b'(r-j,t-jk+i).    
\end{align*}

\end{proof}
\label{rmk4}



\subsection{Solving the Recursion}


Similar to the expansion method presented in \cite{Liron}, the computation of $b'(r,t)$ proceeds by recursive expansion until all $b'$ expressions reach the boundary conditions and vanish. The final sum accumulates positive contributions only from the additive constant $1$ in Case B. This constant from the term $1+\sum\limits_{i=1}^{k-1}\sum\limits_{j=1}^{q-1}b'(r-j,t-jk+i)$ is added once whenever a  $b'(\tilde{r},\tilde{t})$ satisfying the Case B condition $k(\tilde{r}-1) \leq \tilde{t} \leq k \tilde{r} -1$ is expanded.  In the expansion of $b'(r,t)$, these instances are parameterized by the $t+1$ integer pairs $(\tilde{r}, \tilde{t})$, where $\tilde{r}= \lfloor \frac{\tilde{t}}{k} \rfloor+1$ for $\tilde{t}=0,1, \dots,t$.  Specifically, we aim to count the number of times  that $b'(\tilde{r}, \tilde{t})$ appears in the complete expansion of $b'(r,t)$. We define the set of admissible pairs as
\begin{align*}
    \mathcal{S}_k^{(q)} = \{(q, (q-1)k),&(1,0),(1,1), \dots, (1,k-1),(2,k),(2,k+1), \dots, (2,2k-1), \dots,\\& (q-1, (q-2)k), (q-1, (q-2)k+1), \dots, (q-1, (q-1)k-1)\}.
\end{align*} According to the recursion given in Lemma \ref{lm14}, the problem is equivalent to counting the number of possible  sequences of pairs $\big((r_1,t_1), (r_2, t_2), \dots, (r_m,t_m)\big)$ of arbitrary length $m \ge 1$, where each   $(r_j,t_j)$ is  chosen from $\mathcal{S}_k^{(q)}$, such that $\sum_{j=1}^m r_j = r-\tilde{r}$ and $\sum_{j=1}^m t_j = t-\tilde{t}$. Here,  each pair $(r_j,t_j)$ represents the decrements in the indices $r$ and $t$, respectively.

\begin{definition}
         Let $\# Q(\Delta r, \Delta t)$ denote the number of possible sequences of  pairs $\big((r_1,t_1), \dots, (r_m,t_m)\big)$ for any length $m \ge 1$, where each $(r_j,t_j) \in \mathcal{S}_{k}^{(q)}$  such that $\sum_{j=1}^m r_j=\Delta r$ and $\sum_{j=1}^m t_j = \Delta t$. Furthermore, let $\# Q_l (\Delta r, \Delta t)$ denote the number of such sequences in which the pair $(q,(q-1)k)$ occurs exactly $l$ times.
\end{definition}    
We start by computing $\# Q_0(\Delta r, \Delta t)$.

\begin{lemma}\cite{Liron} \label{lmP}
    Given integers $\Delta r$ and $\Delta t$, let $\#P_0(\Delta r, \Delta t)$ denote the number of sequences  of $\Delta r$ pairs $\big((r_1,t_1), \dots,$\\$ (r_{\Delta r}, t_{\Delta r}) \big)$ chosen from  $\mathcal{S}_k^{(2)} \setminus \{(2,k)\}$ such that $\sum_{j=1}^{\Delta r} r_j= \Delta r$ and $\sum_{j=1}^{\Delta r} t_j = \Delta t$. Equivalently, this is the number of integer sequences $(t_1, \dots, t_{\Delta r})$ satisfying  $\sum_{j=1}^{\Delta r} t_j = \Delta t$ and $0 \le t_j \le k-1$. Then 
    \begin{align*}
        \# P_0(\Delta r, \Delta t) = \sum_{i=0}^{\lfloor \frac{\Delta t}{k}\rfloor} (-1)^i \binom{\Delta r}{i}  \binom{\Delta r + \Delta t - ik-1}{\Delta r - 1}.
    \end{align*}
\end{lemma}

\begin{lemma}\label{lmQ0}
    Let $ q, k, \Delta r$, $\Delta t$ be given nonnegative integers, with $q \geq 2$. Define $\mathcal{W}_{ q,k, \Delta r, \Delta t}$ to be the set of all sequences $\big((z_1,v_1), (z_2,v_2), \dots, (z_{q-1},v_{q-1})\big)$ of nonnegative integer pairs  satisfying the following conditions:
    \begin{itemize}
        \item $\sum_{i=1}^{q-1} i z_i = \Delta r$,
        \item $\sum_{i=1}^{q-1} v_i = \Delta t$,
        \item  $(i-1)k z_i \leq v_i\leq (ik-1)z_i$  for all $i \in [q-1]$.
    \end{itemize} Then
    \begin{align*}
        \#Q_0(\Delta r, \Delta t)=\sum_{\substack{\big((z_i,v_i)\big)_{i=1}^{q-1}\in \\\mathcal{W}_{q,k,\Delta r, \Delta t}}} \frac{(\sum_{i=1}^{q-1} z_i)!}{\prod_{i=1}^{q-1}z_i!}\prod_{i=1}^{q-1} \#P_0(z_i, v_i-(i-1)kz_i).
    \end{align*}
\end{lemma}
\begin{proof}
For each $i \in [q-1]$, let $z_i$ be the number of pairs with first coordinate $i$ chosen from $\mathcal{S}_k^{(q)}$.  For a fixed $i$, choosing a sequence of $z_i$ such pairs is equivalent to choosing an integer sequence $(y_1^{(i)}, y_2^{(i)}, \dots, y_{z_i}^{(i)})$ where each $y_j^{(i)} \in \{0\} \cup [k-1]$. Here, the second coordinate of the $j$-th pair is defined as $(i-1)k+y_j^{(i)}$. Thus, the sum of the first coordinates of these pairs is $iz_i$, and the corresponding sum of their second coordinates is given by 
\begin{align}
v_i &= \sum_{j=1}^{z_i} \big((i-1)k + y_j^{(i)} \big) \notag \\
&= (i-1)kz_i + \sum_{j=1}^{z_i} y_j^{(i)}, \label{eqVi}
\end{align}
where $0 \leq y_j^{(i)} \leq k-1$ for all $j \in [z_i]$. This mapping ensures that the number of ways to choose the sequence of pairs for  fixed $z_i$ and $v_i$ is equivalent to counting the integer sequences $(y_j^{(i)})_{j=1}^{z_i}$ satisfying $\sum_{j=1}^{z_i} y_j^{(i)} = v_i - (i-1)kz_i$ and $0 \le y_j^{(i)} \le k-1$, which is given by Lemma \ref{lmP} as
    $\# P_0(z_i, v_i - (i-1)kz_i)$.  Since pairs with different first coordinates are distinct, for a fixed $\big((z_1, v_1), (z_2,v_2),\dots,(z_{q-1},v_{q-1})\big)$, the total number of ways to choose the sequence of all pairs  is  \[\frac{(\sum_{i=1}^{q-1} z_i)!}{\prod_{i=1}^{q-1}z_i!}\prod_{i=1}^{q-1} \#P_0(z_i, v_i-(i-1)kz_i).\]

By the definition of $Q_0(\Delta r, \Delta t)$, we have
\begin{equation}\label{eqMain}
\begin{aligned}
    \sum_{i=1}^{q-1} iz_i&= \Delta r,\\ \sum_{i=1}^{q-1} v_i &= \Delta t .
\end{aligned}
\end{equation}
Futhermore, we determine the admissible range for $v_i$. Since $0 \le \sum_{j=1}^{z_i} y_j^{(i)} \leq (k-1)z_i$,  it follows from Equation  \eqref{eqVi} that for each $i \in [q-1]$,
\begin{align}
(i-1)kz_i \leq v_i \leq (ik-1)z_i. \label{eqCompare}
\end{align}
Combining the constraints in \eqref{eqMain} with the ranges in \eqref{eqCompare}, we observe that the collection of all valid sequences of pairs $\big((z_i, v_i)\big)_{i=1}^{q-1}$ is precisely the set $\mathcal{W}_{q,k,\Delta r, \Delta t}$ as defined in Lemma \ref{lmQ0}. Summing the number of solutions over all such sequences in $\mathcal{W}_{q,k,\Delta r, \Delta t}$ yields the desired expression for $\#Q_0(\Delta r, \Delta t)$. This completes the proof.  
\end{proof}

 We now consider the computation of  $\#Q(\Delta r, \Delta t)$.
    
\begin{lemma}\label{lmQ}
    \begin{align*}
        \#Q(\Delta r, \Delta t)=\sum_{l=0}^{\lfloor \frac{\Delta t}{(q-1)k} \rfloor} \sum_{\substack{\big((z_i,v_i)\big)_{i=1}^{q-1} \in\\ \mathcal{W}_{q,k,\Delta r-ql, \Delta t - (q-1)kl}}} \binom{l+\sum_{i=1}^{q-1} z_i}{l}  \frac{(\sum_{i=1}^{q-1} z_i)!}{\prod_{i=1}^{q-1}z_i!} \prod_{i=1}^{q-1} \#P_0(z_i, v_i-(i-1)kz_i).
    \end{align*}
    \end{lemma}

\begin{proof}
    We begin by calculating $\#Q_l(\Delta r, \Delta t)$. If we first choose $l$ times the pair $(q,(q-1)k)$, we are left with $\#Q_0(\Delta r -ql, \Delta t -(q-1)kl)$ ways to select the remaining pairs. For a fixed sequence $\big((z_1,v_1),(z_2,v_2),\dots, (z_{q-1}, v_{q-1})\big) \in \\ \mathcal{W}_{q,k,\Delta r-ql, \Delta t - (q-1)kl}$,  we then have $\binom{l+\sum_{i=1}^{q-1} z_i}{l}$ ways to insert the $l$ copies of $(q,(q-1)k)$ inside the remaining pairs, and thus
    \[
    \#Q_l(\Delta r, \Delta t)=\sum_{\substack{\big((z_i,v_i)\big)_{i=1}^{q-1} \in\\ \mathcal{W}_{q,k,\Delta r-ql, \Delta t - (q-1)kl}}} \binom{l+\sum_{i=1}^{q-1} z_i}{l} \frac{(\sum_{i=1}^{q-1} z_i)!}{\prod_{i=1}^{q-1}z_i!} \prod_{i=1}^{q-1} \#P_0(z_i, v_i-(i-1)kz_i).
    \]
Summing over all possible values of $l$ yields the asserted result. 
\end{proof}

\begin{lemma}\label{lmb'}
    \begin{align*}
        b'(r,t) = \sum_{j=0}^{t} \#Q(r-\lfloor \frac{j}{k} \rfloor-1, t-j).
    \end{align*}
\end{lemma}
\begin{proof}
    As discussed above,  expanding $b'(r,t)$ reaches exactly $t+1$ pairs $(\tilde{r},\tilde{t})$ that satisfy $k(\tilde{r}-1) \leq \tilde{t} \leq k\tilde{r}$ and  $0 \leq \tilde{t} \leq t$ , each of which contributes to the sum. These are precisely the  pairs $(r_j, t_j) = (\lfloor \frac{j}{k} \rfloor+1,j)$ for $0 \leq j \leq t$, and each  is reached in exactly $\#Q(r-r_j,t-t_j)$ ways. Summing over all these contributions yields the result.
\end{proof}
\begin{corollary}\label{Corcomputeb}
    The combined results of  Lemmas \ref{lm13}, \ref{lm14},  \ref{lmP}, \ref{lmQ0}, \ref{lmQ}, and  \ref{lmb'} give an explicit expression for $|D_t(B_{r,k;q})|$. We restate the results here:  \\
When $q \geq r$,
    \begin{align*}
        b(r,t)&=
\begin{cases}
0, &\text{if $t<0$ or $t>kr$;}~\\
\sum\limits_{i=0}^{t}b(r-1,i), &\text{if $0\leq t\leq k$;}~\\
\sum\limits_{i=0}^{k}b(r-1,t-i), &\text{otherwise.}
\end{cases}
    \end{align*}
 When $q <r$,
 \begin{align*}
     b(r,t)=b'(r,t)+\sum_{i=1}^{q-1} b'(r-i,t-ik),
 \end{align*} 
 where  \begin{align*}
        b'(r,t) &= \sum_{j=0}^{t} \#Q(r-\lfloor \frac{j}{k} \rfloor-1, t-j),\\
        \#Q(\Delta r, \Delta t)&=\sum_{l=0}^{\lfloor \frac{\Delta t}{(q-1)k} \rfloor} \sum_{\substack{ \big((z_i,v_i)\big)_{i=1}^{q-1} \in \\ \mathbb{W}_{q,k,\Delta r-ql, \Delta t - (q-1)kl}}} \binom{l+\sum_{i=1}^{q-1} z_i}{l}  \frac{(\sum_{i=1}^{q-1} z_i)!}{\prod_{i=1}^{q-1}z_i!} \prod_{i=1}^{q-1} \#P_0(z_i, v_i-(i-1)kz_i),\\
\# P_0(\Delta r, \Delta t)& = \sum_{i=0}^{\lfloor \frac{\Delta t}{k}\rfloor} (-1)^i \binom{\Delta r}{i}  \binom{\Delta r + \Delta t - ik-1}{\Delta r - 1},
    \end{align*}
    and 
    \begin{align*}
     \mathcal{W}_{q,k,\Delta r-ql, \Delta t - (q-1)kl} =
\Big\{&\big((z_i,v_i)\big)_{i=1}^{q-1}:
\sum_{i=1}^{q-1}i z_i=\Delta r-ql, \sum_{i=1}^{q-1} v_i= \Delta t-(q-1)kl,\\& (i-1)kz_i \leq v_i \leq (ik-1)z_i, z_i,v_i \in \mathbb{Z}_{\geq 0}\Big\}.
    \end{align*}
\end{corollary}

\begin{remark}
    When $q=2$, the expression for $|D_t(B_{r,k;q})|$ coincides with that given in Corollary  $5$ of \cite{Liron}. 
\end{remark}

We present an algorithm to determine the set $\mathcal{W}_{q,k,\Delta r, \Delta t}$, which proceeds in two stages via dynamic programming.
\begin{itemize}
\item For each integer $ \frac{k\Delta r - \Delta t}{k} \leq x \leq  k\Delta r - \Delta t$,  we first find all nonnegative integer sequences $Z=(z_1, z_2, \dots, z_{q-1})$ satisfying:
\begin{align*}
    \sum_{i=1}^{q-1} i z_i = \Delta r ~\text{and}~ \sum_{i=1}^{q-1} z_i = x.
\end{align*} Let $\mathcal{U}_{x}$ be the set of these solutions and define the total solution set $\mathcal{U}$ as the union over all admissible integer $x$, that is, $\mathcal{U}= \cup_{x \in I} \mathcal{U}_x$, where $I = \{x \in \mathbb{Z}: \frac{k \Delta r -\Delta t}{k} \leq x \leq k \Delta r -\Delta t\}$.
\item Then, for each such sequence $Z$, we determine all nonnegative integer sequences $V=(v_1, v_2, \dots, v_{q-1})$ such that:
\begin{align*}
\sum_{i=1}^{q-1} v_i = \Delta t,
\end{align*}
subject to the constraints $(i-1)k z_i \leq v_i \leq (ik -1)z_i$ for all $i \in [q-1]$.  The set of these solutions is denoted by $\mathcal{V}_Z$.
\end{itemize}

Accordingly, our desired set is given as follows:
\begin{align*}
    \mathcal{W}_{q,k,\Delta r, \Delta t}=\Big\{\big((z_i,v_i)\big)_{i=1}^{q-1}:Z=(z_1, z_2,\dots,z_{q-1}) \in \mathcal{U}, V =(v_1,v_2,\dots,v_{q-1}) \in \mathcal{V}_{Z} \Big\}.
\end{align*}

Using balanced strings, we have derived upper bounds on the number of distinct subsequences in general strings. A comparison between our bound in Corollary \ref{Corcomputeb} and previous bounds is shown in Figure 1. 

\begin{figure}[htbp]
\centering 
\includegraphics[height=6cm]{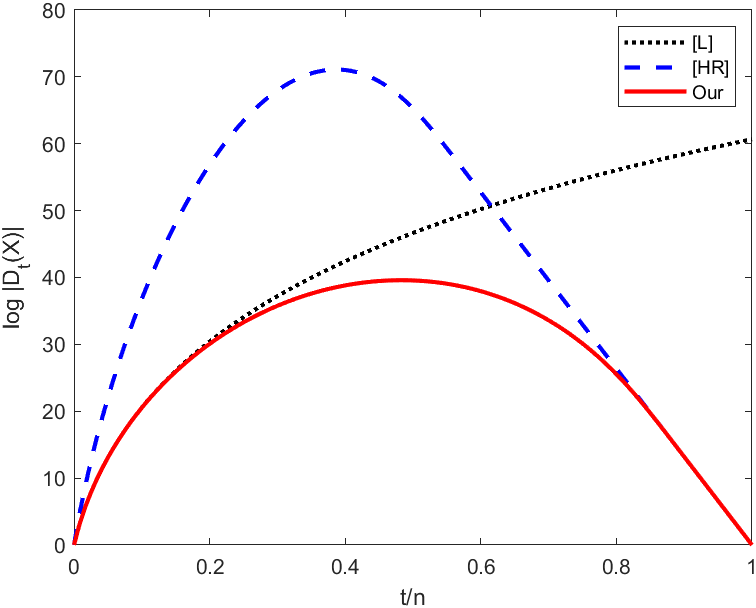}\caption{Comparison of upper bounds for the case $q=3$, $n=120$, $r=24$. Our upper bounds derived from balanced strings [Corollary \ref{Corcomputeb}] are compared against previous best known bounds. [L] marks the upper bound proven by Levenstein \cite{L0}. [HR] marks the upper bound proven by Hirschberg and Regnier \cite{Hirschberg}.  The results are presented as functions of $t$ on a logarithmic scale.}\label{fig:example}
\end{figure} 

\begin{theorem}
 For fixed $q$, $b(r,k,t;q)$ can be computed in polynomial time in $n$.
\end{theorem}
\begin{proof}
 In the following analysis, we assume that for $m$ bit integers, an addition operation takes time $O(m)$, and a multiplication takes time $O(m^2)$.  Using dynamic programming, calculating $\binom{a}{b}$ requires $O(ab)$ such basic operations. For the calculation of $b(r,k,t;q)$, all required binomial coefficients include $a$ and $b$ that are of size $O(n)$. Dynamically building a table of all needed binomial coefficients requires $O(n^2)$  additions, and time $O(n^3)$.

We begin by analyzing the time complexity of the above algorithm for computing $\mathcal{W}_{q,k,\Delta r, \Delta t}$. 
\begin{itemize}
    \item In the first step, for a fixed $x$,  dynamically building a table of the state space requires time $O(\frac{qnx^2}{k}\cdot\frac{(x+q-2)^{q-2}}{(q-2)!})$ because each state stores all feasible partial solutions that satisfy its constraints. Since this process is repeated for all possible values of $x$ (ranging from $\frac{k\Delta r -\Delta t}{k}$ to $k\Delta r -\Delta t$), the overall time complexity for the first step becomes $ O(\frac{qn(n-t)^3}{k}\cdot \frac{(n-t+q)^{q-2}}{(q-2)!})$.
    \item Similarly, in the second step, for a fixed sequence $Z$, dynamically building a table of the state space requires time $O(qtn\frac{(t+q-2)^{q-2}}{(q-2)!})$.  Since this process is repeated for all possible sequences $Z$ (at most $\binom{n-t+q-2}{q-2}$), the overall time  complexity for the second step is  $O(qtn\frac{(t+q)^{q-2}(n-t+q)^{q-2}}{((q-2)!)^2})$.
\end{itemize}
  Hence, the overall time comlexity for the entire algorithm is $A:=O(\frac{qn(n-t)^3}{k}\cdot\frac{(n-t+q)^{q-2}}{(q-2)!}+qtn\frac{(t+q)^{q-2}(n-t+q)^{q-2}}{((q-2)!)^2})$.

 Now we count the calculations needed.
\begin{itemize}
    \item $\#P_0(\Delta r, \Delta t)$ requires at most $\frac{2t}{k}$ values of binomial coefficients, and $\frac{2t}{k}$ multiplication and addition operations, taking time $O(n^2\frac{t}{k})$.
    \item For a fixed $l$ and a fixed sequence $\big((z_i,v_i)\big)_{i=1}^{q-1} \in \mathcal{W}_{q,k,\Delta r-ql,\Delta t -(q-1)kl}$, the computation of the corresponding term in $\#Q(\Delta r, \Delta t)$ requires $q-1$ calls to $\#P_0$,  one value of binomial coefficients, and at most $2(q+\sum_{i=1}^{q-1}z_i)$  multiplication operations, taking time $B:=O((q\frac{t}{k}+2q+2n-2t)n^2)$. Then $\#Q(\Delta r, \Delta t)$ takes time $C:=\frac{t}{qk}\big(A+B(n-t)\frac{[(n-t+q)(t+q-2)]^{q-2}}{((q-2)!)^2}\big)$.
    \item $b'(r,k,t;q)$ requires $t+1$ calls to $\#Q$ and $t$ additions, taking time $O(tC)$.
\end{itemize}

All together, when $q < r$, computing $b(r,k,t;q)$ takes time $O(D):=O(n^3+\frac{t^2}{k}\big(\frac{qn(n-t)^3}{k}\cdot\frac{(n-t+q)^{q-2}}{(q-2)!}+qtn\frac{(t+q)^{q-2}(n-t+q)^{q-2}}{((q-2)!)^2}+ (q\frac{t}{k}+2q+2n-2t)n^2(n-t)\frac{[(n-t+q)(t+q-2)]^{q-2}}{((q-2)!)^2}\big));$ when $q \geq r$, computing $b(r,k,t;q)$ takes time $O(\max\{t,k\}(q-\frac{n}{k})D)$.
\end{proof}
\subsection{Asymptotic analysis}
Following the approach in \cite{Liron}, we introduce  deletion patterns for $q$-ary strings in order to   quantify the multiplicative gap between our upper bound $|D_t(B_{r,k;q})|$ and the previous bounds.
\begin{definition}
Let $X$ be a string  with $X=S(x_1, \dots, x_r;a_1,\dots, a_r)$. A deletion pattern of size $t$  for $X$ is an $r$-tuple $(y_1,\dots,y_r)$ of nonnegative integers such that
$\sum_{i=1}^r y_i=t$
 and  $y_i \in [0,x_i]$ for all $0 \leq i \leq r$. Here, $y_i$ denotes the number of characters deleted from the $i$-th run of $X$. For example, applying the deletion pattern $(1,0,2) $ to the string $2220011111$ yields  the subeseqence $2200111$. Let $\mathcal{P}_t(X)$ denote the set of all deletion patterns of size $t$ for X.
\end{definition}
Note that distinct deletion patterns for a string may result in the same subsequence.  For example,  for the string $020011$, the deletion patterns $(1,1,0,0)$ and $(0,1,1,0)$ both yield the subsequence $0011$. Since the binary balanced strings defined in \cite{Liron} and our $q$-ary balanced strings share a similar structure,  we establish analogous results in the following. Specifically, by adapting the techniques used in the proofs of Lemma 20 and Corollary 7 in \cite{Liron}, we obtain Lemma \ref{RfL1} and Corollary \ref{RFL2}.
\begin{lemma}\label{RfL1}
    \begin{align*}
        |D_t(B_{r,k;q})| \le |\mathcal{P}_t(B_{r,k;q})|,
    \end{align*}
    where $|\mathcal{P}_t(B_{r,k;q})| = \sum_{i=0}^{\lfloor \frac{t}{k+1} \rfloor} (-1)^i \binom{r}{i} \binom{r+t-i(k+1)-1}{r-1}$.
    
\end{lemma}

\begin{corollary}\label{RFL2}
    \[
    |D_t(B_{r,k;q})| \leq  \min \left \{\binom{r+t-1}{t},(k+1)^r \right\}.
    \]
\end{corollary}

\begin{lemma}
    For sufficiently small $\epsilon > 0$, let $t=n(\frac{1}{2}-\epsilon)$ and $n=rk$. Then the upper bound of  $\min \left \{\binom{r+t-1}{t},(k+1)^r \right\}$ improves upon the previous bounds $\sum_{i=0}^t \binom{n-t}{i} D_{q-1}(t,t-i)$ and $\binom{r+t-1}{t}$ by a multiplicative factor of at least $\frac{1}{12k\sqrt{r}}c^r$, for some constant $c >1$.
\end{lemma}
\begin{proof}
    Our result follows from the proofs of Lemmas 22 and 23 in \cite{Liron}, together with the observation that $\sum_{i=0}^t \binom{n-t}{i} D_{q-1}(t,t-i)  \ge \sum_{i=0}^t \binom{n-t}{i}$ for   $q \ge 2$.
\end{proof}

Overall, we demonstrate that for values $t$ close to $n/2$ and for sufficiently large $r$ and $k$, $|\mathcal{P}_t(B_{r,k;q})|$ (and hence our upper bound of $|D_t(B_{r,k;q})|$) improves on the bounds from  \cite{Hirschberg} and \cite{L0} by an exponential multiplicative factor of $2^{\Omega(r)}$. Our results are illustrated in Figure \ref{fig:example2}.

\begin{figure}[htbp]
\centering 
\includegraphics[height=6cm]{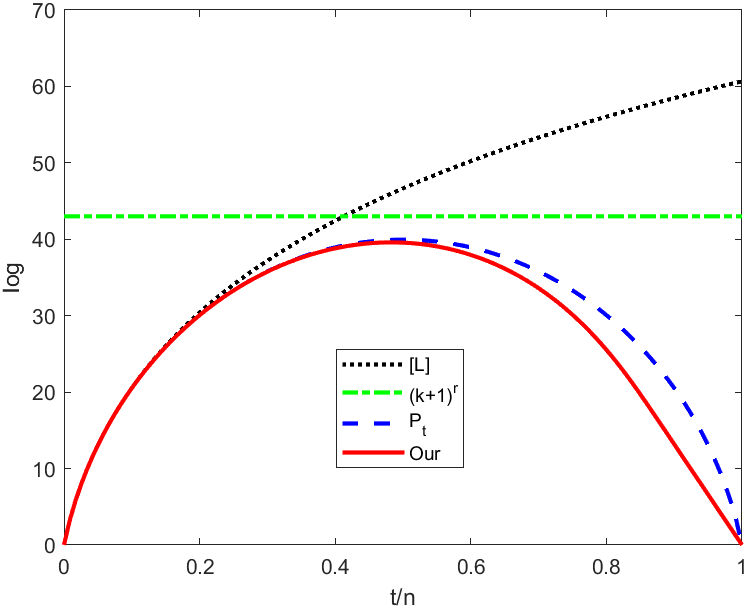}\caption{Our upper bound for $|D_t(B_{r,k;q})|$ from Corollary \ref{Corcomputeb} and the bound $\mathcal{P}_t(B_{r,k;q})$ compared to the upper bounds of [2] (marked [L]), and to the combinatorial bound $(k+1)^r$. Our example is computed for $n=120$ and $r=24$ as a function of $t$ (in logarithmic scale).}\label{fig:example2}
\end{figure} 

\section{Conclusion}
\label{sec6}
In this paper, we study the number of subsequences obtained from $q$-ary deletion channels, where $q\geq 2$. For our lower bound, we use the reduction operation to derive a lower bound for general $q$-ary strings from the known lower bound for binary unbalanced strings. For our upper bound, we first use the cyclic operation on any $q$-ary string $U=S(x_1,\ldots,x_r;a_1,\ldots,a_r)$ to obtain the block cyclic $q$-ary string $S(x_1,\ldots,x_r)=S(x_1,\ldots,x_r;0,1,\ldots,r~\bmod q)$ which is a generalization of  the binary string $S(x_1,\ldots,x_r;0,1,\ldots,r~\bmod 2)$. 
Further, we extend the balancing operation on $S(x_1,\ldots,x_r;0,1,\ldots,r\bmod 2)$, and prove that the balanced $r$-run string $B_{r,k;q}$ has the maximum number of subsequences among all  $q$-ary strings of length $n=rk$ with $r$ runs. Furthermore, we derive a universal upper bound for arbitrary $n$ and $r$, which significantly improves the best-known results. Moreover, we give a recursive expression for the exact number of subsequences of a balanced string $B_{r,k;q}$ and present a closed-form formula for its evaluation. Finally, through an analysis of deletion patterns, we establish an asymptotic improvement over previous bounds.  While our work concentrates on the scenario where $r \mid n$, the complete characterization of extremal strings for the case $r\nmid n$  remains open.

\section*{Acknowledgment}
The work of X.~Wang is supported by the National Natural Science Foundation of China (Grant No. 12001134). The work of F.-W.~Fu is supported by the National Key Research and Development Program of China (Grant No. 2022YFA1005000), the National Natural Science Foundation of China (Grant No.  62371259),  the Fundamental Research Funds for the Central Universities of China (Nankai University), and the Nankai Zhide Foundation.


\begin{thebibliography}{1}
\bibliographystyle{IEEEtran} 



\bibitem{Hirschberg}
D. S. Hirschberg and M. Regnier, ``Tight bounds on the number of string subsequences,'' \textit{J. Discrete Algorithms}, vol. 1, no. 1, pp. 123-132, 2000. \par

\bibitem{Zhang}
D. Zhang, G. Ge, and Y. Zhang, ``Sequence Reconstruction over 3-Deletion Channels,''In \textit{Proc. Int. Symp. Inform. Theory}, pp. 891-896, 2024. \par

\bibitem{Mercier1}
H. Mercier, M. Khabbazian, and V. K. Bhargava, ``On the number of subsequences when deleting symbols from a string,'' \textit{IEEE Trans. Inf. Theory}, vol. 54, no. 7, pp. 3279-3285, Jul. 2008.\par

\bibitem{Mercier2} 
H. Mercier, V. K. Bhargava, and V. Tarokh, ``A survey of error-correcting codes for channels with symbol synchronization errors,'' \textit{IEEE Commun. Surveys Tuts.}, vol. 12, no. 1, pp. 87-96, Feb. 2010.\par

\bibitem{Kash} 
I. A. Kash, M. Mitzenmacher, J. Thaler, and J. Ullman, ``On the zero-error capacity threshold for deletion channels,'' in \textit{Proc. Inf. Theory Appl. Workshop (ITA)}, 2011, pp. 1-5.\par

\bibitem{Calabi}
L. Calabi and W. E. Hartnett, ``Some general results of coding theory with applications to the study of codes for the correction of synchronization errors,'' \textit{Inf. Control}, vol. 15, no. 3, pp. 235-249, Sep. 1969.\par

\bibitem{Mitzenmacher}
M. Mitzenmacher, ``A survey of results for deletion channels and related synchronization channels,'' \textit{Probab. Surv.}, vol. 6, pp. 1-33, 2009.\par

\bibitem{Gabrys}
R. Gabrys and E. Yaakobi, ``Sequence reconstruction over the deletion channel,'' \textit{IEEE Trans. Inf. Theory}, vol. 64, no. 4, pp. 2924-2931, 2018.\par

\bibitem{L0}
V. I. Levenshtein, ``Binary codes capable of correcting deletions, insertions, and reversals,'' \textit{Soviet Phys. Doklady}, vol. 10, no. 8, pp. 707-710, 1966.\par

\bibitem{L1}
V. I. Levenshtein, ``Efficient reconstruction of sequences from their subsequences or supersequences,'' \textit{J. Combin. Theory, Ser. A}, vol. 93, no.2, pp. 310-332, 2001.\par

\bibitem{L2}
V. I. Levenshtein, ``Efficient reconstruction of sequences,'' \textit{IEEE Trans. Inf. Theory}, vol. 47, no. 1, pp. 2-22, 2001.\par

\bibitem{Pham1}
V. L. P. Pham, K. Goyal, and H. M. Kiah, ``Sequence reconstruction problem for deletion channels: a complete asymptotic solution,'' In \textit{Proc. IEEE Int. Symp. Inform. Theory}, pp. 992-997, 2022.\par

\bibitem{Pham2}
V. L. P. Pham, K. Goyal, and H. M. Kiah, “Sequence reconstruction problem for deletion channels: A complete asymptotic solution,” \textit{J. Combinat. Theory A}, vol. 211, p. 105980,  2025.\par

\bibitem{Liron}
Y. Liron and M. Langberg, ``A characterization of the number of subsequences obtained via the deletion channel,'' \textit{IEEE Trans. Inf. Theory}, vol. 61, no. 5, pp. 2300-2312, 2015.\par

\bibitem{Sun}
Y. Sun, Y. Xi, and G. Ge, ``Sequence reconstruction under single-burst insertion/deletion/edit channel,'' \textit{IEEE Trans. Inf. Theory}, vol. 69, no. 7, pp. 4466-4483, 2023.

\bibitem{Wang}
G. Wang and Q. Wang, ``On the size distribution of the fixed-length Levenshtein balls with radius one,'' \textit{Designs, Codes and Cryptography}, vol. 92, pp. 2253-2265, 2024.

\bibitem{Lan}
Z. Lan, Y. Sun, W. Yu and G. Ge, ``Sequence reconstruction under channels with multiple bursts of insertions or deletions," \emph{IEEE Trans. Inf.  Theory}, vol. 72, no. 1, pp. 315-330, 2026. 

\bibitem{Bar}
D. Bar-Lev, T. Etzion, and E. Yaakobi, ``On the size of balls and anticodes of small diameter under the fixed-length Levenshtein metric,'' \textit{IEEE Trans. Inf. Theory}, vol. 69, no. 4, pp. 2324-2340, 2023.

\bibitem{He}
L. He and M. Ye, ``The size of Levenshtein ball with radius 2: expectation and concentration bound,'' in Proceeding of the \textit{International Symposium on Information Theory (ISIT)}, Taipei, Taiwan, 2023, pp. 850-855.


\end{thebibliography}
\end{document}